\documentclass[aps,pra,10pt,superscriptaddress,twocolumn,floatfix]{revtex4-2}
\usepackage{graphicx}
\usepackage{amsmath,amsthm,amssymb,dsfont}
\usepackage{xcolor}
\usepackage{soul}
\soulregister{\cite}{7}\soulregister{\cref}{7}\soulregister{\Cref}{7}\soulregister{\ref}{7}
\usepackage{orcidlink}
\usepackage{tikz}
\newcommand{\ket}[1]{\left| #1 \right\rangle}
\newcommand{\bra}[1]{\left\langle #1 \right|}
\newcommand{\proj}[1]{\ket{#1}\hskip-1mm\bra{#1}}
\newcommand{\av}[1]{\langle #1\rangle}
\newcommand{\braket}[2]{\langle #1|#2 \rangle}

\usepackage{hyperref}
\usepackage[capitalise]{cleveref}
\crefname{equation}{Eq.}{Eqs.}
\crefname{figure}{Fig.}{Figs.}
\crefname{section}{Sec.}{Secs.}
\Crefname{section}{Sec.}{Secs.}
\crefname{appendix}{App.}{Apps.}
\Crefname{appendix}{App.}{Apps.}
\crefname{table}{Table}{Tables}
\newcommand{\idop}{\hat{\mathds{1}}}
\newtheorem{proposition}{Proposition}
\usepackage[authormarkuptext=name,commentmarkup=uwave]{changes}
\definechangesauthor[name={JRH},color=orange!70!black]{Jonte}

\begin{document}

\title{Frame dependence of Kochen-Specker contextuality for relativistic spin systems}
\author{Jonte R. Hance\,\orcidlink{0000-0001-8587-7618}}
\email{jonte.hance@newcastle.ac.uk}
\affiliation{Quantum Group, School of Computing, Newcastle University, 1 Science Square, Newcastle upon Tyne, NE4 5TG, UK}
\author{Ruben Campos Delgado\, \orcidlink{0000-0002-9957-3604}}
\email{ruben.camposdelgado@itp.uni-hannover.de}
\affiliation{Institut f\"ur Theoretische Physik, Leibniz Universit\"at Hannover, 30167 Hannover, Germany}

\begin{abstract}
We ask whether Kochen-Specker contextuality is Lorentz invariant. On the full Hilbert space, it is, since boosts act unitarily. Restricted to spin, it is not: momentum acts as an environment applying momentum-dependent Wigner rotations, leaving spin observables unsharp and breaking the operator identities behind state-independent contextuality. For spin $1/2$ one scalar governs the loss, the packet average $\Delta$ of $\sin^2(\theta_W/2)$ (where $\theta_W$ is the Wigner angle); for any spin, a ray degrades at a rate set by its spin variance transverse to the boost. In the spin-only account, the Free Will Theorem's SPIN axiom fails for every $\Delta>0$; the Yu-Oh set loses contextuality between $\Delta\approx0.008$ and $0.016$, with state independence going first; the Peres-Mermin square loses contextuality at $\Delta\approx0.059$ or $0.074$ by orientation; singlet CHSH loses nonlocality for $\Delta\geq0.159$; and a GHZ-state OR gate loses Raussendorf's contextuality certificate at $\Delta\approx0.206$. A $q$-plate implements the channel exactly on photon polarisation, with a retardation in place of the rapidity, so the nonlocality, GHZ and aimed-source results are testable without motion.
\end{abstract}

\maketitle

\section{Introduction}\label{sec:intro}

The Kochen-Specker (KS) theorem shows that the outcomes of quantum measurements cannot be assigned independently of the measurement context: for Hilbert spaces of dimension three or more, no map from self-adjoint operators to values can respect the functional relations between commuting observables \cite{kochen1967problem,bell1966problem,budroni2022kochen}. This form of contextuality is a property of the observables alone. In its state-independent form, every two-qubit state violates the Peres-Mermin noncontextuality inequality by the same amount \cite{peres1990incompatible,mermin1990simple,cabello2008experimentally}. KS contextuality therefore differs~\cite{bozzetto2026warringcontextualitiesprovably} from the operational notion introduced by Spekkens \cite{spekkens2005contextuality}, which asks whether the probabilities assigned to preparations and (possibly unsharp) measurements admit a noncontextual ontological model.

One of us recently showed that Spekkens' contextuality is, in general, not Lorentz invariant when restricted to spin degrees of freedom~\cite{camposdelgado2026frame}: a set of spin states which is contextual in one inertial frame can be noncontextual in another, because the boost entangles spin with momentum, and tracing out the momentum turns pure, linearly dependent spin states into mixed, linearly independent ones. Given that mechanism affects the states, one might not expect it to change KS contextuality, which no choice of states can turn off. We therefore ask here whether KS contextuality is frame dependent as well, and through what mechanism.

On the full Hilbert space of spin and momentum together, KS contextuality is Lorentz invariant for the same reason as Spekkens' contextuality~\cite{camposdelgado2026frame}: boosts act unitarily, and unitaries preserve both the spectral structure of observables and every functional relation between them. However, we can also ask what happens when observers restrict their attention to spin. For both Spekkens and KS contextuality, the channel formed by the boost followed by the momentum trace does affect the contextuality of the system; what differs is which object the notion is defined on. In the Schr\"odinger picture the channel mixes the states, which is the mechanism causing the change in Spekkens contextuality; in the Heisenberg picture, the channel turns the sharp spin observables of one frame into unsharp observables in the other, which causes the change in KS contextuality. The operator identities between the nine Peres-Mermin observables, which the KS argument needs, fail for their unsharp images, by an amount set by the momentum spread of the wave packets and the rapidity of the boost.

Two earlier discussions inform our result. One is the finite-precision debate around the KS theorem. Meyer, Kent and Clifton~\cite{meyer1999finite,kent1999noncontextual,clifton2000simulating} built noncontextual models reproducing quantum statistics whenever the specification of which sharp observable is measured carries a finite error; Barrett and Kent~\cite{barrett2004non} defend those models while judging the loophole they open logically possible but physically implausible; Cabello \cite{cabello2002finite} argues that finite-precision unsharpness does not touch the theorem's physical content, and Spekkens \cite{spekkens2014status} traces much of the difficulty to assuming outcome determinism rather than deriving it. That debate concerns imprecision in the choice of a sharp observable. Ours concerns observables that are genuinely unsharp, their effects not projectors, for which the relevant treatment is Krishna \emph{et al.}'s robust noncontextuality inequalities~\cite{krishna2017deriving}; boosts provide a physical mechanism generating such unsharpness, set by rapidity and packet width. The second discussion leads from our previous work: one of us (with collaborators) recently argued that the elements of a positive operator-valued measure describe correlations between a system and quantum fluctuations of its environment, which select the measurement context outcome by outcome~\cite{hance2026external}. Boosted spin measurements realise this directly: the particle's own momentum wave function acts as the environment, with the Wigner rotation being the momentum-controlled unitary, meaning the rapidity sets the coupling.

Relativistic effects on spin correlations are well studied for nonlocality. Peres, Scudo and Terno showed that the reduced spin state of a single particle is not covariant, so its entropy is not a relativistic invariant~\cite{peres2002quantum,peres2004quantum}; spin-spin entanglement transfers to and from the momentum sector under boosts~\cite{gingrich2002quantum,gingrich2003entangled}; and the degree of Bell violation depends on the observer's motion~\cite{czachor1997einstein,terashima2003einstein,terashima2003relativistic}. Terashima and Ueda showed that, for a sharp momentum, the Wigner rotation is a rotation the observer can undo by rotating their settings~\cite{terashima2003relativistic}; a genuine loss therefore needs a momentum spread, which our channel is built on. For contextuality, Chen \emph{et al.} showed that a relativistic hydrogen atom violates a CHSH-type noncontextuality inequality on its four-dimensional spinor space, within a fixed frame; for a free Dirac electron and a Dirac oscillator the same algebra violates it, but the observables mix positive- and negative-energy states and are not measurable~\cite{chen2013quantum}. The non-relativistic analogue of the thresholds below was recently mapped out by Dutra \emph{et al.} for a global depolarising channel~\cite{dutra2026contextuality}; ours acts per particle, so the numbers differ. To our knowledge, the behaviour of KS contextuality under a change of inertial frame has not been analysed, and Ref.~\cite{camposdelgado2026frame} is the nearest point of comparison. (Generalised contextuality has separately been extended to continuous-variable systems \cite{jokinen2026generalised}, but without reference to frames.)

\Cref{sec:channel} sets up the boost-induced spin channel and fixes measurement conventions, and \cref{sec:invariance} shows KS contextuality is Lorentz-invariant on the full Hilbert space. \Cref{sec:pm} works out the boosted Peres-Mermin square: the state-independent sequential value $\chi(\Delta)$, its threshold, which orientations of the square keep commuting contexts, and the success probability of a GHZ-state measurement-based computation. \Cref{sec:arbitrary} generalises the channel to arbitrary spin in closed form, shows that every ray degrades at a rate set by its planar spin variance, quantifies the fragility of the Yu-Oh set, and asks what survives of the Free Will Theorem's SPIN axiom. \Cref{sec:photons} models massless particles, where collinear boosts leave polarisation contextuality untouched and transverse boosts dephase it. \Cref{sec:bell} shows the consequences of our results for Bell tests and device-independent certification, and \cref{sec:aimed} shows that momentum-correlated sources relocate the frame in which certification is ideal, for massive particles and, near-implementably, for photons. \Cref{sec:emulation} shows a $q$-plate can emulate the massive channel on photon polarisation, so that the nonlocality and aimed-source results can be tested with a retardation in place of a rapidity. \Cref{sec:context} interprets the mechanism as context selection by the momentum environment, and \cref{sec:discussion} closes with experimental scales and open directions.

\section{Relativistic qubits and the boost channel}\label{sec:channel}

We use the standard description of massive particles with spin \cite{wigner1939unitary,weinberg1995quantum}, in the conventions of Ref.~\cite{camposdelgado2026frame}. A particle of mass $m$ has four-momentum $p^\mu=(E,\mathbf{p})$ with $E=\sqrt{m^2+|\mathbf{p}|^2}$, and momentum eigenstates normalised as $\braket{p}{p'}=(2\pi)^3\,2E\,\delta^3(\mathbf{p}-\mathbf{p}')$, with invariant measure $d\mu(p)=d^3\mathbf{p}/[(2\pi)^3 2E]$. Under a Lorentz transformation $\Lambda$,
\begin{equation}
\hat{U}(\Lambda)\ket{\sigma,p}=\sum_{\sigma'}\hat{D}^{(1/2)}_{\sigma'\sigma}[W(\Lambda,p)]\ket{\sigma',\Lambda p},
\label{eq:translaw}
\end{equation}
where
\begin{equation}\label{eq:def_wigner}
W(\Lambda, p) = L^{-1}(\Lambda p)\,\Lambda\, L(p)
\end{equation}
is the Wigner rotation, with $L(p)$ the standard boost taking the rest four-momentum $(m,\mathbf{0})$ to $(E,\mathbf{p})$. For a boost with rapidity $\zeta$ along $z$, and momentum in spherical coordinates $(p,\theta,\varphi)$, the spin-$1/2$ representation is \cite{Halpern1968,gingrich2002quantum,camposdelgado2026frame}
\begin{equation}
\hat{D}^{(1/2)}=\begin{pmatrix}c_W & s_W e^{-i\varphi}\\ -s_W e^{i\varphi} & c_W\end{pmatrix},
\label{eq:wignermatrix}
\end{equation}
with $c_W=\cos(\theta_W/2)$ and $s_W=\sin(\theta_W/2)$ given by
\begin{align}
c_W&=\sqrt{\frac{E+m}{E'+m}}\left[\cosh\left(\frac{\zeta}{2}\right)+\frac{p\cos\theta}{E+m}\sinh\left(\frac{\zeta}{2}\right)\right],\label{eq:alpha}\\
s_W&=\frac{p\sin\theta}{\sqrt{(E+m)(E'+m)}}\sinh\left(\frac{\zeta}{2}\right),\label{eq:beta}
\end{align}
and $E'=E\cosh\zeta+p\cos\theta\sinh\zeta$, so that $c_W^2+s_W^2=1$. \Cref{eq:wignermatrix} is thus a rotation by the Wigner angle $\theta_W$ about $\hat{n}=(\sin\varphi,-\cos\varphi,0)\propto\hat{p}\times\hat{z}$, orthogonal to both the boost and the momentum:
\begin{equation}
\hat{D}^{(1/2)}=\exp\left[-i\tfrac{\theta_W}{2}\,\hat{n}\cdot\boldsymbol{\sigma}\right].
\label{eq:rotform}
\end{equation}

For a particle prepared in a product of a pure spin state $\hat{\rho}=\proj{\chi}$ and a momentum wave packet, $\ket{\psi}=\int d\mu(p)\,\psi(p)\ket{\chi}\otimes\ket{p}$ with $\int d\mu(p)|\psi(p)|^2=1$, applying \cref{eq:translaw} and tracing out the momentum gives the spin state seen by the boosted observer, 
\begin{equation}
\Lambda_\zeta(\hat{\rho})=\int d\mu(p)\,|\psi(p)|^2\,\hat{D}(p)\hat{\rho}\hat{D}(p)^\dagger ,
\label{eq:channel}
\end{equation}
a random-rotation channel: trace preserving, unital, and self-dual whenever the packet is azimuthally symmetric, since $\hat{n}(\varphi+\pi)=-\hat{n}(\varphi)$ at fixed $\theta_W$ then pairs each rotation with its inverse at equal weight. We take $|\psi(p)|^2\propto e^{-p^2/2\sigma^2}$, a Gaussian of width $\sigma$ centred on the rest frame of the devices.

Averaging the channel \cref{eq:channel} over the azimuthal angle of $\hat{n}$ leaves it diagonal in the Pauli basis,
\begin{equation}
\Lambda_\zeta(\hat{\sigma}_z)=\eta_z\,\hat{\sigma}_z,\qquad
\Lambda_\zeta(\hat{\sigma}_{x,y})=\eta_\perp\,\hat{\sigma}_{x,y},
\label{eq:etas}
\end{equation}
with both coefficients fixed by a single scalar,
\begin{gather}
\eta_z=1-2\Delta,\qquad \eta_\perp=1-\Delta,\nonumber\\
\Delta=\int d\mu(p)\,|\psi(p)|^2\,\sin^2(\theta_W/2) .
\label{eq:Delta}
\end{gather}
This packet average of $\sin^2(\theta_W/2)$ is the only quantity through which the boost enters \cref{sec:pm,sec:bell}'s spin-$1/2$ calculations. (Spin $s$ needs the higher moments $\Delta_j=\langle\sin^{2j}(\theta_W/2)\rangle$ up to $j=2s$, \cref{sec:arbitrary}, and the massless case has no $\Delta$ at all, \cref{sec:photons}.) $\Delta$ vanishes at $\zeta=0$, grows monotonically with rapidity, and saturates as $\zeta\to\infty$ at
\begin{equation}
\Delta_\infty=\int d\mu(p)\,|\psi(p)|^2\,\frac{p^2\sin^2\theta}{2(E+m)(E+p\cos\theta)}
\;\approx\;\frac{\sigma^2}{2m^2}
\label{eq:Deltainf}
\end{equation}
for $\sigma\ll m$. Churilov characterised the spin channels reachable at a fixed boost for inversion-symmetric momentum distributions, a family that contains the Pauli-diagonal members $\mathrm{diag}(1-a,1-a,1-2a)$~\cite{churilov2026fixed}; spherically symmetric packets land on those members with $a=\Delta$, so $\eta_z=2\eta_\perp-1$. Narrow wave packets thus bound the effect of any boost, however fast the observer moves: the spin sector is degraded by relativistic momentum spread, not by relativistic velocity.

A remark on conventions. Which operator a detector on a moving particle measures is not settled: candidate relativistic spin operators (Foldy-Wouthuysen, Pryce, and Pauli-Lubanski based constructions, among others) differ for wave packets~\cite{bauke2014relativistic}, and the operational reading of Wigner rotations in Bell-type experiments has been questioned~\cite{saldanha2012physical}. We adopt the convention of Ref.~\cite{camposdelgado2026frame}: states are expanded in the Wigner basis, a measurement counts as sharp in the frame where its device is at rest, and \cref{eq:translaw} describes the system in every other frame. Our results compare descriptions across frames within this one convention. The competing operator choices coincide for narrow packets at rest and differ at the same order in momentum spread that generates $\Delta$. An $\mathcal{O}(\sigma^2/m^2)$ ambiguity in the operator is thus the size of the whole effect, so the channel coefficients, and every width and rapidity quoted below, are convention-dependent even though the mechanism, momentum-conditioned rotation followed by a trace, is not.

\section{Invariance on the full Hilbert space}\label{sec:invariance}

\begin{proposition}\label{prop:invariance}
Let $\{\hat{A}_i\}$ be a set of self-adjoint operators on $\mathcal{H}=\mathcal{H}_{\rm spin}\otimes\mathcal{H}_{\rm momentum}$ admitting no noncontextual value assignment (respecting the functional relations between commuting members). Then for any Lorentz transformation $\Lambda$, the set $\{\hat{U}(\Lambda)\hat{A}_i\hat{U}(\Lambda)^\dagger\}$ admits no noncontextual value assignment either, and the two sets give identical measurement statistics on correspondingly transformed states.
\end{proposition}

The proof of this comes from the fact $\hat{U}(\Lambda)$ is unitary, so it preserves commutation relations, spectra, and every polynomial identity among the $\hat{A}_i$, and a value assignment for the transformed set would pull back to one for the original set. The invariance of the statistics is the invariance of transition probabilities under \cref{eq:translaw}. The corresponding statement holds for Spekkens' contextuality when spin and momentum are kept together~\cite{camposdelgado2026frame}.

This has two consequences. First, KS contextuality of the total system is a Lorentz-invariant notion. Second, any observer running the Peres-Mermin protocol with devices at rest in their own frame obtains the ideal value, whatever their motion relative to the source, since the proposition applies to their comoving observables. Frame dependence can only appear in the relation between one observer's spin sector and another's devices.

\section{The Peres-Mermin square for a boosted observer}\label{sec:pm}

\emph{The square at rest---}
Consider two spin-$1/2$ particles, each in a product of a spin state and a Gaussian wave packet as in \cref{sec:channel}, with devices at rest in the lab frame. The Peres-Mermin square \cite{peres1990incompatible,mermin1990simple} is the array of nine two-qubit observables
\begin{equation}
\begin{array}{ccc}
\hat{A}_{11}=\hat{\sigma}_z\otimes\idop & \hat{A}_{12}=\idop\otimes\hat{\sigma}_z & \hat{A}_{13}=\hat{\sigma}_z\otimes\hat{\sigma}_z\\
\hat{A}_{21}=\idop\otimes\hat{\sigma}_x & \hat{A}_{22}=\hat{\sigma}_x\otimes\idop & \hat{A}_{23}=\hat{\sigma}_x\otimes\hat{\sigma}_x\\
\hat{A}_{31}=\hat{\sigma}_z\otimes\hat{\sigma}_x & \hat{A}_{32}=\hat{\sigma}_x\otimes\hat{\sigma}_z & \hat{A}_{33}=\hat{\sigma}_y\otimes\hat{\sigma}_y
\end{array}
\label{eq:square}
\end{equation}
Here each row $\hat{R}_i$ and each column $\hat{C}_j$ is a triple of mutually commuting dichotomic observables; the product of each triple is $+\idop$, except the third column, whose product is $-\idop$. Measuring the triples sequentially and summing the products of outcomes,
\begin{equation}
\chi=\av{\hat{R}_1}+\av{\hat{R}_2}+\av{\hat{R}_3}+\av{\hat{C}_1}+\av{\hat{C}_2}-\av{\hat{C}_3},
\label{eq:chi}
\end{equation}
quantum mechanics gives $\chi=6$ for every two-qubit state, while any noncontextual assignment of outcomes obeys $\chi\leq 4$ \cite{cabello2008experimentally,guhne2010compatibility}. That bound is the one for deterministic noncontextual assignments to dichotomic procedures. The boosted square could equally be tested against Krishna \emph{et al.}'s robust generalised-noncontextuality inequalities~\cite{krishna2017deriving}, which do not presuppose sharpness; we leave that comparison open.

\emph{The square seen from a moving frame---}
Now describe the same nine measurements from a frame boosted with rapidity $\zeta$ along $z$. As per \cref{prop:invariance}, over the full Hilbert space nothing changes. Restricted to spin, each sharp effect $\hat{\Pi}_\pm=(\idop\pm \hat{A})/2$ of the lab devices becomes correlated with momentum through the Wigner rotation, and conditioning on the momentum wave functions (in the sense of Ref.~\cite{hance2026external}) yields the unsharp spin effects
\begin{equation}
\tilde{\Pi}_\pm=\tfrac{1}{2}\left(\idop\pm c_A\hat{A}\right),\qquad
\tilde{A}=(\Lambda_\zeta^\dagger\otimes\Lambda_\zeta^\dagger)(\hat{A})=c_A\hat{A} ,
\label{eq:noisyeffects}
\end{equation}
where $c_A$ is the product of one factor of $\eta_z$ or $\eta_\perp$, from \cref{eq:etas}, for each nontrivial Pauli factor of $\hat{A}$. The nine boosted observables retain their eigenbases, their pairwise commutation, and the labels of the square, but lose their sharpness, and with it the algebra. The row and column products are no longer $\pm\idop$:
\begin{equation}
\tilde{A}_{11}\tilde{A}_{12}\tilde{A}_{13}=\eta_z^4\,\idop\neq\idop ,
\end{equation}
and similarly for the other five triples (\cref{fig:square}). The KS argument, which multiplies the six identities together to reach $+1=-1$, no longer holds, and not because the products have shrunk: an assignment $v(\tilde{A}_{ij})=\pm c_{ij}$ obeying multiplicativity on commuting triples would meet the same parity contradiction. What fails instead is the warrant for such an assignment. The boosted effects are not projectors, the recorded outcomes $\pm1$ are not eigenvalues of $\tilde{A}_{ij}$, and outcome determinism, which the KS argument assumes, has nothing to attach to \cite{spekkens2014status}. This is the situation Krishna \emph{et al.} analyse for a square of unsharp observables \cite{krishna2017deriving}; theirs arises from imperfect devices, ours from the kinematics of the observer.

\begin{figure}[tb]
\centering
\begin{tikzpicture}[every node/.style={font=\small}]
\begin{scope}
\foreach \i in {0,...,3} \draw[gray!60] (0.85*\i,0) -- (0.85*\i,2.1);
\foreach \j in {0,...,3} \draw[gray!60] (0,0.7*\j) -- (2.55,0.7*\j);
\draw[thick] (0,0) rectangle (2.55,2.1);
\node at (0.425,1.75) {$\hat{Z}\idop$};
\node at (1.275,1.75) {$\idop \hat{Z}$};
\node at (2.125,1.75) {$\hat{Z}\hat{Z}$};
\node at (0.425,1.05) {$\idop \hat{X}$};
\node at (1.275,1.05) {$\hat{X}\idop$};
\node at (2.125,1.05) {$\hat{X}\hat{X}$};
\node at (0.425,0.35) {$\hat{Z}\hat{X}$};
\node at (1.275,0.35) {$\hat{X}\hat{Z}$};
\node at (2.125,0.35) {$\hat{Y}\hat{Y}$};
\node[anchor=west] at (2.7,1.75) {$+\idop$};
\node[anchor=west] at (2.7,1.05) {$+\idop$};
\node[anchor=west] at (2.7,0.35) {$+\idop$};
\node at (0.425,-0.35) {$+\idop$};
\node at (1.275,-0.35) {$+\idop$};
\node at (2.125,-0.35) {$-\idop$};
\node at (1.7,2.5) {rest frame};
\end{scope}
\begin{scope}[yshift=-3.6cm]
\foreach \i in {0,...,3} \draw[gray!60] (1.55*\i,0) -- (1.55*\i,2.1);
\foreach \j in {0,...,3} \draw[gray!60] (0,0.7*\j) -- (4.65,0.7*\j);
\draw[thick] (0,0) rectangle (4.65,2.1);
\node at (0.775,1.75) {$\eta_z\, \hat{Z}\idop$};
\node at (2.325,1.75) {$\eta_z\, \idop \hat{Z}$};
\node at (3.875,1.75) {$\eta_z^2\, \hat{Z}\hat{Z}$};
\node at (0.775,1.05) {$\eta_\perp\, \idop \hat{X}$};
\node at (2.325,1.05) {$\eta_\perp\, \hat{X}\idop$};
\node at (3.875,1.05) {$\eta_\perp^2\, \hat{X}\hat{X}$};
\node at (0.775,0.35) {$\eta_z\eta_\perp\, \hat{Z}\hat{X}$};
\node at (2.325,0.35) {$\eta_\perp\eta_z\, \hat{X}\hat{Z}$};
\node at (3.875,0.35) {$\eta_\perp^2\, \hat{Y}\hat{Y}$};
\node[anchor=west] at (4.8,1.75) {$+\eta_z^4\,\idop$};
\node[anchor=west] at (4.8,1.05) {$+\eta_\perp^4\,\idop$};
\node[anchor=west] at (4.8,0.35) {$+\eta_z^2\eta_\perp^4\,\idop$};
\node at (0.775,-0.35) {$+\eta_z^2\eta_\perp^2\,\idop$};
\node at (2.325,-0.35) {$+\eta_z^2\eta_\perp^2\,\idop$};
\node at (3.875,-0.35) {$-\eta_z^2\eta_\perp^4\,\idop$};
\node at (2.8,2.5) {boosted frame};
\end{scope}
\end{tikzpicture}
\caption{The Peres-Mermin square described at rest relative to the devices (top) and by a boosted observer after tracing out momentum (bottom), with $\hat{Z}\idop$ shorthand for $\hat{\sigma}_z\otimes\idop$ and so on. The operator products of each triple appear beside its row and below its column; in the boosted frame they acquire factors of $\eta_z=1-2\Delta$ and $\eta_\perp=1-\Delta$, so the contradiction between the six products disappears.}
\label{fig:square}
\end{figure}

\emph{The sequential value---}
To quantify how much contextuality remains, let the boosted observer implement the unsharp observables \cref{eq:noisyeffects} sequentially with L\"uders (square-root) instruments, the standard interpretation of a sequential contextuality test~\cite{guhne2010compatibility}. We assume throughout that a spin measurement leaves the momentum undisturbed, without which a triple could not be measured on one particle. Within one triple the observables commute, so the instruments commute and the correlation factorises,
\begin{equation}
\av{abc}=c_A c_B c_C\,\mathrm{Tr}[\hat{\rho}\,\hat{A}\hat{B}\hat{C}]=\pm\,c_A c_B c_C ,
\label{eq:factorise}
\end{equation}
for every spin state $\hat{\rho}$, so at fixed momentum profile the boosted value inherits the state independence of the ideal test. It is not independent of the momentum-sector state, a freedom \cref{sec:aimed} exploits. Summing the six triples with the signs from \cref{eq:chi},
\begin{equation}
\begin{split}
\chi&(\Delta)=\\
&\left[(1-2\Delta)^2+(1-\Delta)^2\right]^2+2(1-2\Delta)^2(1-\Delta)^4 .
\label{eq:chiDelta}
\end{split}
\end{equation}
This equals $6$ at $\Delta=0$, decreases monotonically on $[0,\tfrac12]$, and crosses the noncontextuality bound $\chi=4$ at
\begin{equation}
\Delta_c\approx 0.058542 ,
\label{eq:Deltac}
\end{equation}
the smallest root of $8\Delta^6-40\Delta^5+107\Delta^4-148\Delta^3+108\Delta^2-40\Delta+2$.

Whether an experiment reaches $\Delta_c$ depends on width and rapidity through $\Delta(\zeta,\sigma/m)$. From \cref{eq:Deltainf}, the saturation value $\Delta_\infty$ exceeds $\Delta_c$ only for
\begin{equation}
\sigma\gtrsim\sigma^*\approx 0.42\,m .
\label{eq:sigmastar}
\end{equation}
For narrower packets the boosted observer sees a reduced but still contextual value at every rapidity; for wider packets there is a finite critical rapidity $\zeta_c(\sigma/m)$ beyond which the sequential L\"uders reading of the lab effects on the boosted spin sector satisfies the Peres-Mermin inequality. Representative values are $\zeta_c\approx 2.82$ at $\sigma=0.5\,m$, $1.80$ at $\sigma=0.7\,m$, and $1.37$ at $\sigma=m$; \cref{fig:chi,fig:critical} show $\chi$ against rapidity and the critical curve.

Beyond the critical curve the maximally mixed state's whole boosted behaviour admits a noncontextual model, because for that state the Peres-Mermin inequality is the only facet of the noncontextual polytope that can be violated (Appendix~\ref{app:ncmodel}). For other spin states the singles and in-context pairs no longer vanish and determining whether they allow a noncontextual model would require a different linear program, which we have not solved.

Consider what $\chi(\Delta)$ means. A lab spin measurement acts as $\hat{\Pi}\otimes\idop_{\rm mom}$, leaving the momentum undisturbed, and the boosted observer describes its two effects at each $p$ by the sharp projectors $\hat{D}(p)\hat{\Pi}_\pm \hat{D}(p)^\dagger$, so the observable itself becomes $\hat{A}(p)=\hat{D}(p)\hat{A}\hat{D}(p)^\dagger$. Measuring a triple on the two-particle state, summing the products of outcomes, and averaging over the packet only at the end, with $\hat{A}(p)$ shorthand for the product of the two single-particle conjugations,
\begin{equation}
\int\! d\mu(p)|\psi(p)|^2\,\mathrm{Tr}\!\left[\hat{\rho}\hat{A}(p)\hat{B}(p)\hat{C}(p)\right]=\pm1 ,
\label{eq:jointaccount}
\end{equation}
because at fixed $p$ a single unitary conjugates the whole context, so $\hat{A}(p)\hat{B}(p)\hat{C}(p)=\hat{D}(p)\hat{A}\hat{B}\hat{C}\hat{D}(p)^\dagger=\pm\idop$ and the operator identities are intact. That account returns $\chi=6$ in every frame, as \cref{prop:invariance} requires, which agrees with the recorded frequencies. \Cref{eq:factorise} provides a different account: it takes the packet average first, at every step, and applies the L\"uders instrument of the resulting unsharp effect $\tilde{\Pi}_\pm$. The two agree on single-shot probabilities but disagree on sequences, because the second discards the spin-momentum correlation between one measurement and the next. $\chi(\Delta)$ belongs to that spin-only account, being the frame dependence of a description, not of an outcome frequency, as are \cref{eq:pfail,eq:aimed} below. \Cref{sec:context} interprets the discarded correlation.

\begin{figure}[tb]
\centering
\includegraphics[width=\columnwidth]{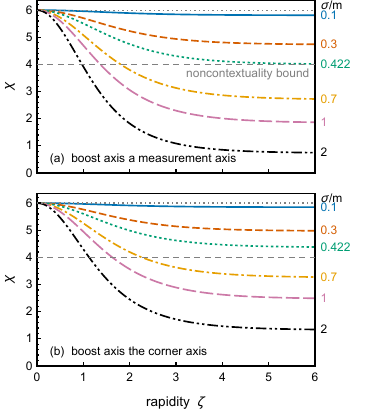}
\caption{Sequential Peres-Mermin value $\chi$ attributed to the lab devices by an observer boosted with rapidity $\zeta$ who considers only the spin sector, for Gaussian wave packets of width $\sigma/m$: (a) the square built so that the boost axis is one of its two measurement axes, \cref{eq:chiDelta}; (b) the square built so that the boost axis appears only in the corner entry and both measurement axes are transverse (\cref{sec:pm}). The dashed line is the noncontextuality bound $\chi=4$; the value at $\zeta=0$ is $6$ for every width, every state, and both orientations.}
\label{fig:chi}
\end{figure}

\begin{figure}[tb]
\centering
\includegraphics[width=\columnwidth]{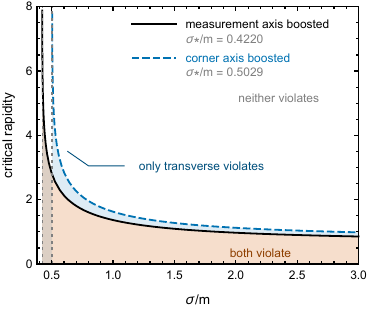}
\caption{Critical rapidity above which the boosted description of the rest-frame Peres-Mermin test satisfies the noncontextuality bound, against wave-packet width, for the two aligned orientations of the square. The curves diverge at $\sigma^*\approx 0.42\,m$ (boost axis as measurement axis, solid) and $\sigma^*_\perp\approx 0.50\,m$ (boost axis as corner axis, dashed). Below both curves the two orientations violate the bound; between them only the transverse square does; above, neither.}
\label{fig:critical}
\end{figure}

The map in \cref{eq:noisyeffects} runs from sharp lab observables to their spin-sector description in the boosted frame. The opposite direction, asking which spin-sector objects in the lab frame reproduce the boosted observer's own sharp measurements, inverts the channel and generally produces operators that are not positive, Ref.~\cite{camposdelgado2026frame}'s pseudo-POVMs. Choice of frame therefore affects which sets of spin measurements certify state-independent contextuality, not whether contextuality exists.

\emph{Orientation of the square---}
The square \cref{eq:square} is built from the orthonormal triad $\{\hat{x},\hat{y},\hat{z}\}$: two directions giving the single-qubit measurement axes, and the third appears only in the corner entry $\hat{A}_{33}$. In the rest frame, every orthonormal triad $\{\hat{a},\hat{b},\hat{c}\}$ possesses an equivalent square. Under a boost the choices diverge. Writing $M=\mathrm{diag}(\eta_\perp,\eta_\perp,\eta_z)$ for the Bloch-space action of the channel, the boosted image of $(\hat{a}\cdot\boldsymbol{\sigma})\otimes(\hat{b}\cdot\boldsymbol{\sigma})$ is $(M\hat{a}\cdot\boldsymbol{\sigma})\otimes(M\hat{b}\cdot\boldsymbol{\sigma})$, and for any vectors $\vec{u},\vec{v}$,
\begin{align}
\big[(\vec{u}\cdot\boldsymbol{\sigma})\otimes(\vec{v}\cdot\boldsymbol{\sigma})&,\,(\vec{v}\cdot\boldsymbol{\sigma})\otimes(\vec{u}\cdot\boldsymbol{\sigma})\big]\nonumber\\
&=2i\,(\vec{u}\cdot\vec{v})\left(X\otimes\idop-\idop\otimes X\right),
\label{eq:commutator}
\end{align}
with $X=(\vec{u}\times\vec{v})\cdot\boldsymbol{\sigma}$, and similar expressions for the other pairs in the mixed triples. The images of every row and column remain mutually commuting iff $M\hat{a}$, $M\hat{b}$, $M\hat{c}$ are pairwise orthogonal, i.e., iff the triad consists of eigenvectors of $M^2$, which for $\eta_z^2\neq\eta_\perp^2$, and so for every $\Delta\in(0,\tfrac12]$, requires it to contain the boost axis. (\Cref{eq:commutator} vanishes for $\vec{u}\parallel\vec{v}$ as well, but that case is excluded by the triad being orthonormal.) A misaligned square not only no longer violates the noncontextuality inequality, but its triples no longer commute, so the sequential L\"uders protocol becomes order-dependent and the outcomes of one run stop defining a single context. (Failing to commute does not by itself make unsharp observables incompatible; heavily damped effects can be jointly measurable without commuting. We therefore lose the repeatable sequential reading the KS argument requires.) The boost therefore selects which context structures remain available, as well as how well they perform (\cref{sec:context}).

This process leaves two aligned orientations. The boost axis can serve as a measurement axis, giving \cref{fig:square}'s square and the value \cref{eq:chiDelta}, or as the corner axis, with both measurement axes transverse, giving
\begin{equation}
\chi_\perp(\Delta)=2(1-\Delta)^4\left[2+(1-2\Delta)^2\right].
\label{eq:chiperp}
\end{equation}
Since $\chi_\perp(\Delta)\geq\chi(\Delta)$ throughout $[0,\tfrac{1}{2}]$, with equality only at $\Delta=0$, a moving observer should lay the square down transverse to their motion. The threshold rises to $\Delta_c^\perp\approx 0.074382$, the smallest root of $4\Delta^6-20\Delta^5+43\Delta^4-52\Delta^3+38\Delta^2-16\Delta+1$, the critical width to $\sigma^*_\perp\approx 0.50\,m$, and the critical rapidities to $\zeta_c^\perp\approx 2.26$ at $\sigma=0.7\,m$ and $1.63$ at $\sigma=m$. At $\sigma=0.5\,m$, where the aligned square loses its violation near $\zeta\approx 2.8$, the transverse square retains it at every rapidity, since $\Delta_\infty(0.5)\approx 0.0738$ sits below $\Delta_c^\perp$, though by under a percent and so within the convention dependence noted in \cref{sec:channel}.

\emph{A GHZ computation---}
The same account can be applied to a computation. In measurement-based quantum computation (MBQC) \cite{Raussendorf:2001zim,Raussendorf:2012tof}, local measurements on an entangled resource state drive the computation, and the simplest example is Anders and Browne's gate on $\lvert\mathrm{GHZ}\rangle=(\lvert 000\rangle+\lvert 111\rangle)/\sqrt2$ \cite{Anders:2009zzs}, read as an OR gate as in Ref.~\cite{Raussendorf:2013wih}. Input bit $i_k$ selects the observable on qubit $k=1,2$, $\hat{X}_k$ for $i_k=0$ and $\hat{Y}_k$ for $i_k=1$, with $\hat{X}_k$, $\hat{Y}_k$ the Pauli operators $\hat{\sigma}_x$, $\hat{\sigma}_y$ on that qubit; qubit $3$ is measured in $\hat{X}_3$ if $i_1\oplus i_2=0$ and in $\hat{Y}_3$ otherwise; and the output is the parity $o=s_1\oplus s_2\oplus s_3$ of the outcomes $s_k\in\{0,1\}$, with $(-1)^{s_k}$ the recorded eigenvalue. Since $\hat{X}\hat{X}\hat{X}$ has eigenvalue $+1$ on $\lvert\mathrm{GHZ}\rangle$ while $\hat{X}\hat{Y}\hat{Y}$, $\hat{Y}\hat{X}\hat{Y}$ and $\hat{Y}\hat{Y}\hat{X}$ have eigenvalue $-1$, the parity equals $i_1\lor i_2$ with certainty for every input, and that certainty is Mermin's proof of the KS theorem \cite{mermin1990simple} read as a computation \cite{Raussendorf:2013wih}. Now let the three qubits be carried by three particles with wave packets as in \cref{sec:channel}, and describe the same devices from a frame boosted along $\hat{z}$, transverse to both measured axes. By \cref{prop:invariance} the lab devices still succeed with probability $1$. In the spin-only account each device is the unsharp effect $\tfrac12(\idop\pm\eta_\perp\hat{O}_k)$ of \cref{eq:noisyeffects}; expanding the product of the three effects, the one- and two-qubit correlators of $\lvert\mathrm{GHZ}\rangle$ vanish and the three-qubit correlator is scaled by $\eta_\perp^3$, so the account attributes to every input the success probability
\begin{equation}
p=\tfrac12\left(1+\eta_\perp^3\right)=\tfrac12\left[1+(1-\Delta)^3\right].
\label{eq:ghzprob}
\end{equation}
Unlike the square, this certificate is state dependent, so \cref{eq:ghzprob} is also what a detector comoving with the boosted observer records when it measures $\hat{X}$ or $\hat{Y}$ sharply on the boosted spin state, as for CHSH in \cref{sec:bell}; and since each particle is measured once, no sequential instrument enters and the value, unlike $\chi(\Delta)$, does not depend on the instrument chosen for the unsharp effects. Raussendorf \cite{Raussendorf:2013wih} showed that an MBQC evaluating a nonlinear Boolean function of $m$ bits with success probability $p_S>1-2^{-m}$ is contextual; his proof bounds the success of any setting-independent value assignment, so it applies to these statistics although his setting is projective. For the OR gate $m=2$, so the computation is certified contextual only for
\begin{equation}
\Delta<1-2^{-1/3}\approx0.2063,
\label{eq:ghzthreshold}
\end{equation}
which holds at every rapidity for $\sigma<\sigma^*_{\rm GHZ}\approx1.381\,m$; above the threshold the bound is silent. For a boost along $\hat{x}$ the four inputs acquire different factors, $\eta_z^3$ for $\hat{X}\hat{X}\hat{X}$ and $\eta_z\eta_\perp^2$ for the other three, and the input-averaged success probability is lower, so the transverse orientation is again the robust one. The threshold sits between \cref{tab:hierarchy}'s two CHSH rows.

\section{Arbitrary spin}\label{sec:arbitrary}

\emph{The channel for spin $s$---}
This analysis generalises to any spin $s$, because the boost acts through the spin-$s$ representation of the same momentum-dependent rotation. The reduced spin channel is \cref{eq:channel} with $\hat{D}^{(s)}$ in place of $\hat{D}^{(1/2)}$. Decompose the operators into irreducible tensor operators $T^k_q$, $k=0,\dots,2s$. Writing $\hat{n}(\varphi)=\hat{R}_z(\varphi')\hat{x}$ with $\varphi'=\varphi-\pi/2$, the rotation is $\hat{R}_z(\varphi')\hat{R}_x(\theta_W)\hat{R}_z(\varphi')^{-1}$. Since $T^k_q$ transforms in the spin-$k$ representation, conjugating by it multiplies the $(q',q)$ element of $\hat{D}^{(k)}$ by $e^{i(q-q')\varphi'}$, so the azimuthal average removes every $q'\neq q$ and leaves $\hat{D}^{(k)}[\hat{R}_x(\theta_W)]_{qq}=d^{\,k}_{qq}(\theta_W)$. The channel is therefore diagonal in the tensor basis,
\begin{equation}
\Lambda^{(s)}_\zeta\left(T^k_q\right)=f^k_q\,T^k_q,\qquad
f^k_q=\left\langle d^{\,k}_{qq}(\theta_W)\right\rangle ,
\label{eq:master}
\end{equation}
the packet average of the diagonal Wigner $d$-functions. For $k=1$ it reproduces \cref{eq:etas}, since $d^{\,1}_{11}=1-\sin^2(\theta_W/2)$ and $d^{\,1}_{00}=1-2\sin^2(\theta_W/2)$; in general $d^{\,k}_{qq}$ is a polynomial of degree $k$ in $\sin^2(\theta_W/2)$, so the spin-$s$ channel is fixed by the moments $\Delta_j=\langle\sin^{2j}(\theta_W/2)\rangle$ with $j\leq 2s$, of which $\Delta_1=\Delta$.

Much of the qubit analysis holds for higher spins. The channel is unital and trace preserving, so the resolutions of identity defining each context are preserved while sharpness is not: for $\Delta>0$ every $f^k_q$ with $k\geq 1$ is strictly below one, and every nontrivial projector maps to an effect with spectrum strictly inside $(0,1)$. The latter requires the channel rather than just the $f^k_q$: $\bra{v}\Lambda^{(s)}_\zeta(P)\ket{v}$ equals $1$ or $0$ only if $\hat{D}(p)^\dagger\ket{v}$ lies in, or orthogonal to, the range of $P$ for almost every $p$, and a packet with continuous support generates too many distinct $\hat{D}(p)$ for either. (A degenerate distribution can fail this: a $\theta_W$ concentrated at $\pi/2$ gives $f^2_0=-1/2$ and sends the $J_z=0$ projector to an effect with a zero eigenvalue.) The KS valuation argument, which assigns $0$ or $1$ to idempotents, therefore does not hold in the boosted frame for any spin. Commutation is again selective: contexts whose projectors are diagonal in the $J_z$ eigenbasis are spanned by the $T^k_0$ and stay mutually commuting under \cref{eq:master}, while generic contexts do not, generalising \cref{sec:pm}.

Contextuality loss also grows quickly with tensor rank. At small Wigner angle,
\begin{equation}
f^k_q=1-\left[k(k+1)-q^2\right]\Delta_1+\mathcal{O}(2),
\label{eq:smallangle}
\end{equation}
writing $\mathcal{O}(2)$ for the terms of joint degree two and above in the moments $\Delta_j$. The $q=0$ component of rank $k$ therefore decays $k(k+1)/2$ times faster than the qubit's $\hat{\sigma}_z$, while a nonzero $q$ slows the decay back by $q^2\Delta_1$. The damping is also not monotone in the Wigner angle: $d^{\,3}_{00}=-0.03$ at $\theta_W=0.7$~rad while $d^{\,1}_{11}=0.88$, so for $s\geq3/2$ a packet concentrated near that angle would reverse the sign of a rank-$3$ component rather than fade it.

\emph{Sharpness loss ray by ray---}
The KS argument works with rays, so the useful measure of the boost's effect is how far each rank-one projector $P_\psi=\proj{\psi}$ is from its boosted image. Its survival $F(\psi)=\mathrm{Tr}[P_\psi\Lambda^{(s)}_\zeta(P_\psi)]$, the weight the boosted effect keeps on the original ray, has a closed form at every order.

\begin{proposition}\label{prop:planar}
For a pure spin-$s$ state $\ket{\psi}$ let $V_\perp(\psi)=\mathrm{Var}_\psi(\hat{J}_x)+\mathrm{Var}_\psi(\hat{J}_y)$ be its planar spin variance transverse to the boost. Then
\begin{gather}
F(\psi)=\left\langle\left|\bra{\psi}\hat{D}^{(s)}\ket{\psi}\right|^2\right\rangle=1-2\Delta_1V_\perp(\psi)+\mathcal{O}(2),
\label{eq:survival}\\
\bra{\chi}\Lambda^{(s)}_\zeta(P_\psi)\ket{\chi}=2\Delta_1\sum_{a=x,y}|\bra{\chi}\hat{J}_a\ket{\psi}|^2+\mathcal{O}(2)
\label{eq:leak}
\end{gather}
for any ray $\ket{\chi}\perp\ket{\psi}$, and, averaged over Haar-random $\ket{\psi}$ in $d=2s+1$ dimensions,
\begin{equation}
\bar{F}=\frac{d+\left\langle\chi_s(\theta_W)^2\right\rangle}{d(d+1)}=1-\frac{d(d-1)}{3}\Delta_1+\mathcal{O}(2),
\label{eq:haar}
\end{equation}
with $\chi_s(\theta)=\sin[(2s+1)\theta/2]/\sin(\theta/2)$ the spin-$s$ character.
\end{proposition}
\begin{proof}
Write $\hat{D}^{(s)}=e^{-i\theta_W\hat{n}\cdot\hat{\mathbf{J}}}$ and expand to second order in $\theta_W$:
\begin{align*}
|\bra{\psi}\hat{D}^{(s)}\ket{\psi}|^2&=1-\theta_W^2\,\mathrm{Var}_\psi(\hat{n}\cdot\hat{\mathbf{J}})+\mathcal{O}(\theta_W^4),\\
\bra{\chi}\hat{D}^{(s)}P_\psi\hat{D}^{(s)\dagger}\ket{\chi}&=\theta_W^2|\bra{\chi}\hat{n}\cdot\hat{\mathbf{J}}\ket{\psi}|^2+\mathcal{O}(\theta_W^4).
\end{align*}
The azimuthal average over $\hat{n}(\varphi)$ in the plane transverse to the boost replaces $\hat{n}\cdot\hat{\mathbf{J}}$ by $\hat{J}_x$ and $\hat{J}_y$ with weight $\tfrac12$ each, and the packet average gives $\langle\theta_W^2\rangle=4\Delta_1+\mathcal{O}(2)$. For \cref{eq:haar}, the channel is a mixture of unitaries $\hat{D}^{(s)}$, so the Haar-averaged survival is the average gate fidelity of that mixture with the identity, $\langle(d+|\mathrm{Tr}\,\hat{D}^{(s)}|^2)\rangle/[d(d+1)]$ \cite{nielsen2002simple}, and $\mathrm{Tr}\,\hat{D}^{(s)}=\chi_s(\theta_W)$; expanding $\chi_s(\theta)^2=d^2-d^2(d^2-1)\theta^2/12+\mathcal{O}(\theta^4)$ gives the first-order form.
\end{proof}

\Cref{eq:survival} says the boost degrades a spin state at a rate set by its spin uncertainty in the plane transverse to the motion, and \cref{eq:leak} that it leaks probability between two orthogonal rays in proportion to the transverse spin matrix element connecting them. For $s=\tfrac12$ both reduce to \cref{eq:etas}: $V_\perp=\tfrac12$ for a $\hat{\sigma}_z$ eigenstate and $\tfrac14$ for a $\hat{\sigma}_x$ eigenstate give $1-F=\Delta$ and $\tfrac12\Delta$, the defects of the effects $\tfrac12(\idop\pm\eta_z\hat{\sigma}_z)$ and $\tfrac12(\idop\pm\eta_\perp\hat{\sigma}_x)$ on their own eigenstates. For the $\hat{J}_z$ eigenstates of spin $s$, $V_\perp(m)=s(s+1)-m^2$, which recovers the spin-$1$ spectra quoted below \cref{eq:ffactors} ($1-F=4\Delta_1$ for $m=0$ and $2\Delta_1$ for $m=\pm1$). Across all states of spin $s$,
\begin{equation}
V_\perp^{\min}(s)\;\leq\;V_\perp(\psi)\;\leq\;s(s+1)-m_0^2 ,
\label{eq:vwindow}
\end{equation}
with $m_0=0$ for integer and $\tfrac12$ for half-integer spin, the upper bound attained by $\ket{s,m_0}$. The lower bound is the planar-squeezing limit of He \emph{et al.}~\cite{he2011planar}: it is $\tfrac{7}{16}$ at $s=1$ (below the spin-coherent value $\tfrac12$) and grows only as $s^{2/3}$, a numerical minimisation gives $V_\perp^{\min}/s^{2/3}$ rising slowly from 0.44 at $s=1$ to 0.54v at $s=16$. The Haar mean $V_\perp=d(d-1)/6$ of \cref{eq:haar} sits far above that.

This replaces the rank-counting heuristic with a statement about thresholds. Every effect of a boosted KS set is the image of a rank-one projector, with sharpness defect $1-F=2\Delta_1V_\perp$ to first order, so \cref{eq:vwindow} bounds how fast any state-independent certificate in dimension $d$ can lose the sharpness the KS argument needs: no faster than $2[s(s+1)-m_0^2]\Delta_1$ per ray, and no slower than $2V_\perp^{\min}(s)\Delta_1$. A ray set drawn at random degrades like $s^2$, and its threshold like $s^{-2}$; a ray set built from planar-squeezed states would degrade only like $s^{2/3}$. Which of these an optimal KS set approaches is the open question: the rays of a KS set are constrained by their orthogonality graph, and the smallest average planar variance compatible with a given graph is a geometric quantity we have not computed. The two cases in hand sit where \cref{eq:vwindow} puts them. The Yu-Oh rays in their standard orientation have $V_\perp$ between $\tfrac49$ and $2$, with mean $1.04$, and the set's thresholds fall $3.6$ times below the qubit pair's; the reorientation by $R$ discussed below lowers the mean over the thirteen rays to $0.81$, in line with the gain in its thresholds. At large spin every entry of \cref{eq:vwindow} grows without bound, so an arbitrarily small momentum spread still suffices at high rapidity, leaving \cref{sec:pm}'s qubit pair as the most frame-robust massive carrier. The first-order coefficients of the sequential pair terms in a certificate such as \cref{eq:yuohL} are built from the same two ingredients, \cref{eq:survival,eq:leak}, together with the square root of the leak operator that the L\"uders instrument of an unsharp effect introduces; we do not need their explicit form here.

\emph{The Yu-Oh set---}
The scaling shows up already at $s=1$, the dimension where the KS theorem was proven. Only $\Delta_1$ and $\Delta_2$ enter into our calculations:
\begin{align}
f^1_0&=1-2\Delta_1, & f^1_{\pm1}&=1-\Delta_1,\nonumber\\
f^2_0&=1-6\Delta_1+6\Delta_2, & f^2_{\pm1}&=1-5\Delta_1+4\Delta_2,\nonumber\\
f^2_{\pm2}&=1-2\Delta_1+\Delta_2. &&
\label{eq:ffactors}
\end{align}
The Yu-Oh construction \cite{yu2012state} uses $13$ rays in $d=3$, dichotomic observables $\hat{A}_i=\idop-2P_i$, and the $24$ orthogonality relations of its exclusivity graph. The combination
\begin{equation}
L=\sum_i\av{\hat{A}_i}-\tfrac{1}{4}\sum_{i\sim j}\av{\hat{A}_i\hat{A}_j},
\label{eq:yuohL}
\end{equation}
with the second sum over the $48$ ordered pairs of orthogonal rays, equals $25/3$ for every qutrit state, against a noncontextual bound of $8$. Thirteen rays is the minimum for any state-independent proof, in any dimension \cite{cabello2016thirteen}.

The boost removes both the set's compatibility structure and its state independence. We take the rays in their usual Cartesian components, read as amplitudes in the $J_z$ eigenbasis, with the boost along $z$. At $\Delta_1=0$ all $24$ edges are commuting pairs; for any $\Delta_1>0$ only $6$ remain: the three pairs within the $J_z$ eigenbasis, and the three within $\{\ket{0},(\ket{{+}1}\pm\ket{{-}1})/\sqrt{2}\}$ (which is \cref{sec:arbitrary}'s Free Will triad). The remaining rays are spread over too many directions for the $J_z$-adapted exception to protect them. For the $18$ edges that no longer commute, the sequential correlator depends on measurement order; the ordered-pair sum in \cref{eq:yuohL} averages the two orders, which differ by less than $10^{-3}$ of $L$ at the thresholds below.

Implementing the boosted effects with square-root instruments, $L$ remains linear in the state, $L(\hat{\rho})=\mathrm{Tr}[\hat{\rho}\,W(\Delta_1,\Delta_2)]$, so its range over states is the interval spanned by the spectrum of $W$. At $\Delta_1=0$ that spectrum is the single value $25/3$; for $\Delta_1>0$ it spreads, and state independence is the first thing the boost removes. This leads to three different violation thresholds,
\begin{equation}
\Delta^{\rm YO}_{\rm SI}\approx0.00818,\;\;
\Delta^{\rm YO}_{\rm mix}\approx0.01044,\;\;
\Delta^{\rm YO}_{\rm best}\approx0.01632,
\label{eq:yuohthreshold}
\end{equation}
below the first of which every qutrit state still violates the bound, above the last of which none does, and with the maximally mixed state crossing in between. Even the largest of the three is $3.59$ times below the qubit-pair $\Delta_c$, in line with the planar variances of its rays (\cref{prop:planar}). All three are quoted along $\Delta_2=\Delta_1^2$; the widths and rapidities use the true packet moments, for which $\Delta_2\approx2\Delta_1^2$. The protected window shrinks to $\sigma^*_{\rm SI}\approx0.132\,m$ for state-independent violation and $\sigma^*_{\rm mix}\approx0.151\,m$ for the maximally mixed state, with critical rapidities $\zeta_c\approx1.22$, $0.62$, $0.51$, $0.38$ at $\sigma/m=0.3$, $0.7$, $1$, $2$ for that state (\cref{fig:yuoh}).

Orientation also has an effect here: rotations about the boost axis commute with the channel and change nothing, but the remaining two Euler angles do. The ray set conjugated by $R=e^{-i(\pi/2)J_y}e^{-i(\pi/2)J_z}$ raises the three thresholds to $0.01041$, $0.01316$ and $0.01639$, and the corresponding widths to $0.151\,m$, $0.171\,m$ and $0.193\,m$. The first two gain about a quarter, the same relative gain the transverse square wins over the aligned one; the third barely moves. Reorientation therefore extends the range of $\Delta_1$ over which the test stays state independent rather than the range over which some state violates, a split the square, being state independent at every $\Delta$, does not have. The square has one advantage here: at $\sigma=0.5\,m$ the transverse square violates at every rapidity while the aligned one does not, whereas above $\sigma=0.151\,m$ neither orientation of the ray set keeps a state-independent violation at any rapidity.

\begin{figure}[tb]
\centering
\includegraphics[width=\columnwidth]{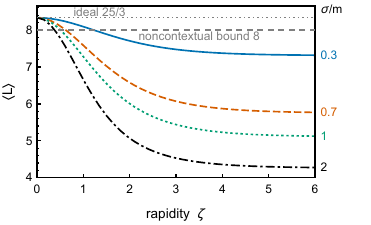}
\caption{Sequential Yu-Oh value for a boosted observer against rapidity, for several wave-packet widths, evaluated on the maximally mixed state with the ray set in its standard orientation. The ideal value is $25/3$ (dotted); the noncontextual bound is $8$ (dashed). Under the boost the value becomes state dependent, spreading by about $\pm0.1$ about these curves near the threshold. Reorienting the ray set raises this threshold by about a quarter (\cref{sec:arbitrary}).}
\label{fig:yuoh}
\end{figure}

\emph{The Free Will Theorem for wave packets---}
The Free Will Theorem \cite{conway2006free} rests on the SPIN axiom: measurements of the squared spin components $(S_x^2,S_y^2,S_z^2)$ of a spin-$1$ particle along any orthogonal triad always yield the outcome pattern $101$ in some order. This is the $d=3$ KS structure, and the spin-$1$ channel \cref{eq:ffactors} determines what remains of it for a boosted observer.

The triad aligned with the boost is an exception to the generic loss of commutation: $S_z^2$ and $S_x^2-S_y^2$ involve only $T^2_0$ and $T^2_{\pm2}$, and $T^2_0$ acts as the identity on the $J_z=\pm1$ subspace, so the boosted images of $S_x^2$, $S_y^2$, $S_z^2$ commute in $(\Delta_1,\Delta_2)$ and the triple experiment stays jointly measurable. However, we lose certainty: the boosted image of $S_z^2$ has eigenvalues $4(\Delta_1-\Delta_2)$ on the $J_z=0$ state and $1-2\Delta_1+2\Delta_2$ on $J_z=\pm1$; the transverse squares are not degenerate in the same way, the boosted $S_x^2$ having spectrum $\{2\Delta_1-\tfrac{3}{2}\Delta_2,\;1-\tfrac{1}{2}\Delta_2,\;1-2\Delta_1+2\Delta_2\}$. Since $\Delta_1^2\leq\Delta_2<\Delta_1$ whenever $\Delta_1>0$, all of these lie strictly inside $(0,1)$, and the $101$ pattern is obtained only with high probability. Expanding branch by branch, the probability of a non-$101$ outcome triple for the maximally mixed state is
\begin{equation}
p_{\mathrm{fail}}=\tfrac{16}{3}\,\Delta_1-\tfrac{14}{3}\,\Delta_2+\mathcal{O}(2),
\label{eq:pfail}
\end{equation}
Since $\Delta_2=\mathcal{O}(\Delta_1^2)$ for a wave packet, the rate in $\Delta_1$ alone is $16/3$: about $0.5\%$ already at $\Delta_1=10^{-3}$ (\cref{fig:pfail}). The full expression is a polynomial of degree three in $(\Delta_1,\Delta_2)$; at $(\Delta_1,\Delta_2)=(0.05,0.004)$ it gives $0.21$ for the maximally mixed state, and optimising over pure states raises it to $0.30$. Rotated triads stop commuting as well. Since the theorem's argument runs through the certainty of the $101$ pattern, the SPIN axiom interpreted as a statement about spin degrees of freedom holds for devices at rest relative to the source and fails, for any wave packet of finite width, in every other frame's spin-only description. A relativistically careful formulation must attach the axiom to the full system of spin and momentum, where \cref{prop:invariance} protects it, or index it to a specified frame. Whether the theorem survives approximate axioms is an existing question \cite{hermens2014conway,cator2014constraints}; the boost makes the imprecision kinematic, fixed by the packet and the rapidity rather than by the apparatus.

\begin{figure}[tb]
\centering
\includegraphics[width=\columnwidth]{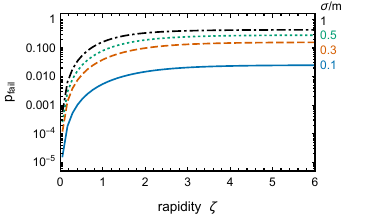}
\caption{Probability of a non-$101$ outcome triple in the spin-only account of the boosted Free Will Theorem experiment (maximally mixed state, triad aligned with the boost), against rapidity, for several wave-packet widths. Read as certainty about spin alone, the SPIN axiom fails at every $\Delta_1>0$; kept on spin and momentum together it does not (\cref{eq:jointaccount}).}

\label{fig:pfail}
\end{figure}

\section{Photons}\label{sec:photons}

\emph{Collinear boosts---}
For a massless particle, $\hat{U}(\Lambda)\ket{p,\lambda}=e^{-i\lambda\xi(\Lambda,p)}\ket{\Lambda p,\lambda}$ with $\lambda=\pm1$ the helicity and $\xi$ the little-group phase, computed from $W(\Lambda,p)=L(\Lambda p)^{-1}\Lambda L(p)$ in the standard massless convention. Gingrich \emph{et al.} record that this phase vanishes for a boost along the beam axis~\cite{gingrich2003entangled}; we give an explicit proof.

\begin{proposition}\label{prop:collinear}
Let $\Lambda$ be a boost along the reference axis $\hat{z}$ of the standard section $L(p)=\hat{R}_z(\varphi)\hat{R}_y(\theta)\hat{B}_z$. Then for every massless $p$ the little-group element $W(\Lambda,p)$ is a null translation, containing no rotation about $\hat{z}$, so $\xi(\Lambda,p)=0$.
\end{proposition}
\begin{proof}
Work in $SL(2,\mathbb{C})$, where the boost is $\Lambda=\mathrm{diag}(e^{\zeta/2},e^{-\zeta/2})$ and $\hat{B}_z(E)=\mathrm{diag}(\sqrt{E/w},\sqrt{w/E})$ carries the reference momentum $k=(w,0,0,w)$ to energy $E$. Since $\Lambda$ commutes with $\hat{R}_z(\varphi)$, the azimuthal factors cancel in $W=L(\Lambda p)^{-1}\Lambda L(p)$, leaving $W=\hat{B}_z(E')^{-1}\hat{R}_y(-\theta')\,\Lambda\,\hat{R}_y(\theta)\hat{B}_z(E)$ with $E'=E(\cosh\zeta+\cos\theta\sinh\zeta)$ the boosted energy and $\theta'$ the boosted polar angle. Multiplying out,
\begin{equation}
W(\Lambda,p)=\begin{pmatrix}1 & -\dfrac{w\sin\theta\sinh\zeta}{E'}\\[6pt] 0 & 1\end{pmatrix}
\label{eq:Wcollinear}
\end{equation}
for every $\theta$, $\varphi$ and $E$. A general element of the little group of $k$ is upper triangular with diagonal $(e^{-i\xi/2},e^{i\xi/2})$, the diagonal carrying the rotation parameter and the off-diagonal entry the null translation, so $\xi=0$. In the geodesic section $L(p)\hat{R}_z(-\varphi)$ of \cref{app:section}, $W$ is conjugated by $\hat{R}_z(\varphi)$, which multiplies the off-diagonal entry by a phase and leaves the diagonal unchanged.
\end{proof}
The same conclusion follows without computation: every factor of $W$ lies in the $SO(2,1)$ subgroup acting on $(t,x,z)$, whose Lie algebra $\mathrm{span}\{J_y,K_x,K_z\}$ meets the little-group algebra of $k$ only along the null generator built from $J_y$ and $K_x$, and $W$ is connected to the identity through boosts of smaller rapidity. The consequence contrasts with the massive case: a polarisation Peres-Mermin square, built on the helicity qubits of two photons, is unaffected by boosts along the beam, for arbitrary wave packets and at every rapidity. Under collinear boosts, then, massless helicity is protected by kinematics rather than by narrow packets. For other boost directions the effect depends on direction, rapidity and angular spread, as it does for polarisation entanglement~\cite{gingrich2003entangled}.

\emph{Transverse boosts---}
A boost transverse to the beam generates a phase, which depends on the section used to fix the polarisation basis. We use the geodesic section $\hat{R}_z(\varphi)\hat{R}_y(\theta)\hat{R}_z(-\varphi)\hat{B}_z$, the minimal rotation carrying $\hat{z}$ to $\hat{p}$: it is the basis in which the transverse field of a paraxial beam is defined, and so the one a lens and waveplate act in, the identification that also underlies Lindner and Terno's account of how focusing degrades a polarisation qubit~\cite{lindner2005focusing}. We do this as the standard massless section gives complete dephasing for every beam width, an artefact of a basis that winds around the beam axis; Appendix~\ref{app:section} shows the difference. For a boost along $x$ and a photon at polar angle $\theta$ with azimuth $\varphi$, the geodesic phase expands as
\begin{equation}
\xi=-\theta\sin\varphi\,t_\zeta+\tfrac{1}{2}\theta^2\sin\varphi\cos\varphi\,t_\zeta^2+\mathcal{O}(\theta^3),
\label{eq:xiparaxial}
\end{equation}
with $t_\zeta=\tanh(\zeta/2)$. Both terms shown, and $\xi$ itself, are odd under $\varphi\to-\varphi$, which is why $\langle\sin2\xi\rangle$ vanishes and the reduced channel is pure dephasing; the rapidity enters only through $\tanh(\zeta/2)$, as it does for the massive Wigner angle. The second-order term makes the residual in \cref{sec:aimed} scale with $\delta^4$. Averaging over a paraxial Gaussian beam of angular spread $\delta$ gives a pure dephasing channel for each helicity qubit, $M=\mathrm{diag}(\eta,\eta,1)$ in the helicity basis, so that the undamped direction is now the third rather than the first two, with
\begin{equation}
\eta=\langle\cos2\xi\rangle,\qquad \eta_\infty\approx1-2\delta^2
\label{eq:eta}
\end{equation}
for small $\delta$ at infinite rapidity (\cref{fig:photoneta}).

\begin{figure}[tb]
\centering
\includegraphics[width=\columnwidth]{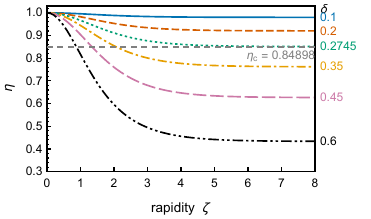}
\caption{Dephasing factor $\eta=\langle\cos2\xi\rangle$ of a photonic polarisation qubit under a boost transverse to the beam, against rapidity, for several beam angular spreads $\delta$. The dashed line is the threshold $\eta_c\approx0.849$ below which the optimally oriented polarisation square satisfies the noncontextuality bound. The curve nearest it, at $\delta=0.2745$, sits just under the critical spread $\delta^*=0.2766$ and so saturates just above $\eta_c$.}
\label{fig:photoneta}
\end{figure}

This channel is the same as for massive particles with the roles of the axes exchanged: there the boost axis was the most damped, here the helicity axis is undamped. \Cref{sec:pm}'s placement rule is unchanged, but now sends the distinguished axis to a measurement position rather than to the corner. The optimal square therefore uses the helicity axis as a measurement axis, giving
\begin{gather}
\chi_{\rm ph}=(1+\eta^2)^2+2\eta^4,\nonumber\\
\eta_c=\left(\frac{\sqrt{10}-1}{3}\right)^{1/2}\approx0.848975,
\label{eq:chiph}
\end{gather}
while the corner-axis variant gives $6\eta^4$ with the worse threshold $(2/3)^{1/4}\approx0.9036$; the gap is $\chi_{\rm ph}-6\eta^4=(1-\eta^2)(1+3\eta^2)\geq0$. Since $\eta$ saturates in rapidity, there is a critical divergence, in analogy with the critical width $\sigma^*$:
\begin{equation}
\delta^*\approx0.2766~\text{rad}.
\label{eq:deltastar}
\end{equation}
Beams narrower than $\delta^*$ keep their polarisation contextuality at every rapidity; wider beams lose it at a finite critical rapidity, e.g., $\zeta_c\approx3.15$, $2.08$, $1.63$, $1.15$ at $\delta=0.3$, $0.35$, $0.4$, $0.5$ (\cref{fig:photoncritical}). An angular spread of $0.28$~rad is a strongly focused beam, comparable to a numerical aperture of $0.3$, so, unlike the Compton-scale localisation the massive thresholds demand, the photonic critical parameter is set by laboratory optics.

\begin{figure}[tb]
\centering
\includegraphics[width=\columnwidth]{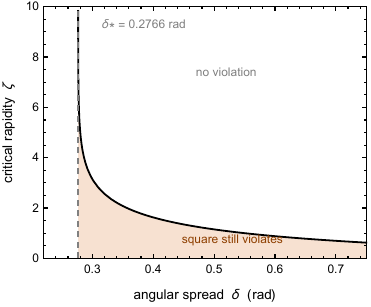}
\caption{Critical rapidity for the transversely boosted polarisation square against beam angular spread $\delta$, diverging at $\delta^*\approx0.277$~rad. Below the curve the square still violates the noncontextuality bound.}
\label{fig:photoncritical}
\end{figure}

The polarisation subspace of a photon is momentum-dependent, by transversality, so identifying a fixed qubit across the packet is itself a convention, of which the section is part; \cref{prop:collinear}, however, holds in either section, since a collinear boost leaves $\varphi$ unchanged. The invariance is a statement about boosts: observer rotations act rigidly on polarisation and change nothing operationally.

\section{Bell tests and device independence}\label{sec:bell}

The frame dependence shown above concerns theoretical attributions to the spin sector, and it should be set against certification schemes whose conclusions are meant to be observer independent. The recorded statistics of any completed experiment are the same in every inertial frame, so a Bell violation, once observed, is not open to relativistic revision. Nonlocality also differs from the contextuality certificates in what the channel does to it: because CHSH is state dependent, a detector with its own comoving analysers records the reduced value, as it does for \cref{sec:pm}'s GHZ gate, where the state-independent square gives $6$ whatever the boost has done to the state. The channel \cref{eq:channel} instead controls the visibility available in the first place when the measured degree of freedom is the spin of massive particles.

\emph{CHSH through the boost channel---}

Send a two-qubit state with correlation tensor $T_{ij}=\mathrm{Tr}[\hat{\rho}\,\hat{\sigma}_i\otimes\hat{\sigma}_j]$ through the channel: $T\mapsto M T M'$ with $M=\mathrm{diag}(\eta_\perp,\eta_\perp,\eta_z)$ for each boosted wing ($M'=\idop$ if only one wing moves). By the Horodecki criterion \cite{horodecki1995violating}, the maximal CHSH value is $2\sqrt{\mu_1^2+\mu_2^2}$ with $\mu_{1,2}$ the two largest singular values of the transformed tensor. For the singlet, $T=-\idop$, and since $\eta_\perp>\eta_z$ the optimal settings lie in the transverse plane, the same orientation rule as \cref{sec:pm}:
\begin{gather}
S_{\rm max}=2\sqrt{2}\,(1-\Delta)^2\quad\text{(both wings)},\nonumber\\
S_{\rm max}=2\sqrt{2}\,(1-\Delta)\quad\text{(one wing)}.
\label{eq:chsh}
\end{gather}
The bound $S=2$ is crossed at $\Delta=1-2^{-1/4}\approx 0.159104$ and $\Delta=1-2^{-1/2}\approx 0.292893$ respectively, corresponding to critical widths $\sigma^*_{\rm CHSH}\approx 1.003\,m$ and $2.487\,m$; representative critical rapidities for the symmetric case are $\zeta_c\approx 3.18$, $2.44$, $1.99$, $1.65$ at $\sigma/m=1.2$, $1.5$, $2$, $3$. At the thresholds where the Peres-Mermin certification fails, the same source still gives $S=2.51$ and $2.42$. The degradation needs the packet to have a spread: at sharp momentum the Wigner rotation is a rotation, and the observer recovers $2\sqrt2$ by rotating their settings \cite{terashima2003relativistic}. \Cref{tab:hierarchy} and \cref{fig:hierarchy} show the resulting hierarchy: the state-independent contextuality certificates fail first, then \cref{sec:pm}'s two-wing CHSH and the GHZ gate, and one-wing CHSH last, with the transverse orientation optimal throughout.

\begin{table}[tb]
\caption{Robustness of each spin-sector certificate to the one-parameter boost channel: the noise threshold $\Delta_c$ at which the certificate fails and the wave-packet width $\sigma^*$ above which some finite rapidity reaches it. The rows are thresholds of the same channel rather than the same experiment: the CHSH and GHZ rows, being state dependent, are what a comoving detector records, the remaining contextuality rows what a spin-only account attributes (\cref{sec:pm}). The GHZ row is Raussendorf's contextuality certificate for the OR gate, \cref{eq:ghzthreshold}. The $\sigma^*$ column inherits the convention dependence of \cref{sec:channel}. The SPIN axiom row refers to certainty of the $101$ outcome pattern, which fails for every $\Delta>0$ (\cref{sec:arbitrary}). The Yu-Oh rows depend on the pair $(\Delta_1,\Delta_2)$ and are quoted along $\Delta_2=\Delta_1^2$ with widths from the true packet moments, and are split because the boost destroys the set's state independence (\cref{sec:arbitrary}); above $\Delta_1=0.01632$ no qutrit state violates. They are for the ray set in its standard orientation, matching \cref{fig:yuoh}; the best orientation raises the thresholds by about a quarter and, since $\sigma^*\propto\sqrt{\Delta}$, the widths by about half that. The photonic square under transverse boosts has the analogous thresholds $\eta_c\approx0.849$ and $\delta^*\approx0.277$~rad (\cref{sec:photons}). The last column is the retardation of a $q=\tfrac12$ plate that emulates the channel at the threshold, $\delta_c=2\arcsin\sqrt{\Delta_c}$ (\cref{sec:emulation}); the qutrit rows have no polarisation emulation.}
\label{tab:hierarchy}
\begin{ruledtabular}
\begin{tabular}{lccc}
Certificate & $\Delta_c$ & $\sigma^*/m$ & $\delta_c$ \\
\colrule
SPIN axiom (as certainty) & $0^+$ & $0^+$ & -- \\
Yu-Oh set ($s=1$), every state & $0.008175$ & $0.1323$ & -- \\
Yu-Oh set ($s=1$), $\hat{\rho}=\idop/3$ & $0.010437$ & $0.1507$ & -- \\
PM square, boost axis measured & $0.058542$ & $0.4220$ & $28.0^\circ$ \\
PM square, transverse & $0.074382$ & $0.5029$ & $31.7^\circ$ \\
CHSH, both wings boosted & $1-2^{-1/4}$ & $1.003$ & $47.0^\circ$ \\
GHZ gate, three particles boosted & $1-2^{-1/3}$ & $1.381$ & $54.0^\circ$ \\
CHSH, one wing boosted & $1-2^{-1/2}$ & $2.487$ & $65.5^\circ$ \\
\end{tabular}
\end{ruledtabular}
\end{table}

\begin{figure}[tb]
\centering
\includegraphics[width=\columnwidth]{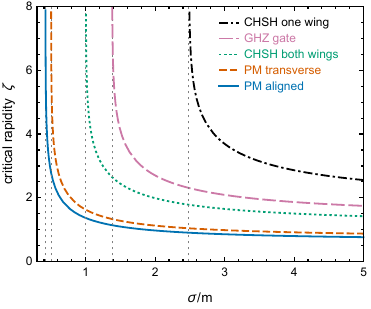}
\caption{Critical curves in the $(\sigma/m,\zeta)$ plane for the five qubit certificates of \cref{tab:hierarchy}; the qutrit ray set, whose thresholds are five to nine times tighter in $\Delta$, is shown separately in \cref{fig:yuoh}. Each certificate holds below its curve, and the ordering matches the table, with one-wing CHSH the most robust.}
\label{fig:hierarchy}
\end{figure}

\emph{Device-independent certification---}
For device-independent protocols \cite{acin2007device,brunner2014bell} the reverse holds. A device-independent certificate, of security, randomness, or a self-tested state, rests only on the observed correlations, which are frame invariant; relative motion can shrink the violation but not invalidate the conclusion drawn from it. Our analysis adds a kinematic noise floor in closed form: the boost contributes the one-parameter family $\mathrm{diag}(1-\Delta,1-\Delta,1-2\Delta)$ of local channels, so a protocol whose analysis requires $S\geq S_0$ tolerates
\begin{equation}
\Delta\leq 1-\left(\frac{S_0}{2\sqrt{2}}\right)^{1/2}
\label{eq:dibudget}
\end{equation}
per wing. Ref.~\cite{acin2007device}'s threshold for collective attacks on a depolarised singlet, a qubit error rate of $7.1\%$, corresponds to $S_0\approx2.427$ and gives $\Delta\leq0.0737$, which sits between \cref{tab:hierarchy}'s two Peres-Mermin rows, though within a percent of the upper one. Noisy preprocessing lowers the requirement further \cite{ho2020noisy}. By contrast, device-dependent arguments resting on a frame-dependent property of the spin sector, such as the contextuality assumptions examined in Ref.~\cite{camposdelgado2026frame}, inherit that frame dependence. In practice the floor sits far below current platforms (\cref{sec:discussion}). The floor itself can be set in the laboratory, on photon polarisation, with the $q$-plate of \cref{sec:emulation}.

\section{Aimed sources}\label{sec:aimed}

So far the source was fixed and the observer moved. The mechanism applies via the momentum wave function, which the source controls.

\emph{Massive particles---}
Let each particle be prepared with a momentum-conditioned twist,
\begin{equation}
\hat{U}_\kappa=\exp\left[+i\tfrac{\kappa}{2}\,\theta_W(p;\zeta^*)\,\hat{n}(\varphi)\cdot\boldsymbol{\sigma}\right],
\label{eq:twist}
\end{equation}
which at $\kappa=1$ is the inverse of the Wigner rotation \cref{eq:rotform} that a boost of rapidity $\zeta^*$ would apply, with $\kappa\in[0,1]$ the compensation strength. Twist and boost rotate about the same momentum-dependent axis $\hat{n}(\varphi)$, so their composition is a rotation by the angle difference, and the whole analysis goes through with one substitution:
\begin{gather}
\chi_{\rm aimed}(\zeta)=\chi(\Delta_{\rm eff}),\nonumber\\
\Delta_{\rm eff}=\left\langle\sin^2\!\left(\frac{\theta_W(\zeta)-\kappa\,\theta_W(\zeta^*)}{2}\right)\right\rangle .
\label{eq:aimed}
\end{gather}
We checked \cref{eq:aimed} against the full conditioned-effect calculation, with the correlated momentum state as Ref.~\cite{hance2026external}'s environment. No choice of $\kappa$ or preparation can exceed the ideal value, since $\Delta_{\rm eff}\in[0,1]$ and $\chi\leq\chi(0)=6$ there. Encodings immune to Wigner rotation have been constructed before, by choosing degrees of freedom the little group acts on trivially \cite{bartlett2005relativistically}; \cref{eq:twist} instead leaves the encoding alone and moves the frame in which it is sharp. Closer to \cref{eq:twist} is Streiter \emph{et al.}'s coherent Wigner rotation \cite{streiter2021relativistic}, a momentum-controlled rotation that restores maximal CHSH violation for momentum superpositions; that acts on the measurement settings and targets the maximum, where the twist acts at the source and targets a nominated frame.

At full compensation, $\kappa=1$, the observer at $\zeta=\zeta^*$ has $\Delta_{\rm eff}=0$ and finds $\chi=6$; the source's rest frame sees $\chi(\Delta(\zeta^*,\sigma/m))$, the value an unaimed source would show at $\zeta^*$. Aiming relocates certification, and can point at one frame only (\cref{fig:aimed}). The trade is asymmetric, because $\theta_W$ saturates through $\tanh(\zeta/2)$ and angle differences shrink at high rapidity: at $\sigma=0.5\,m$ the spin-only account at $\zeta=2$ for a source aimed at $\zeta^*=1$ gives $\chi=5.73$, against $5.38$ in the source's own rest frame. In the language of \cref{sec:context}, \cref{eq:aimed} makes the initial state of the environment a design parameter: the source engineers the momentum fluctuations so that they select sharp contexts for one chosen observer.

\begin{figure}[tb]
\centering
\includegraphics[width=\columnwidth]{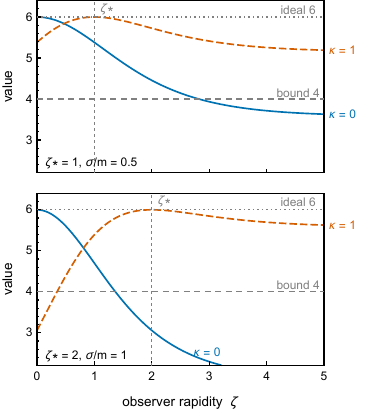}
\caption{Sequential Peres-Mermin value against observer rapidity for an unaimed source ($\kappa=0$) and a source aimed at $\zeta^*$ ($\kappa=1$), for $(\zeta^*=1,\sigma/m=0.5)$ (top) and $(\zeta^*=2,\sigma/m=1)$ (bottom). The aimed source gives the ideal value $6$ at $\zeta=\zeta^*$ and the correspondingly degraded value in its own rest frame.}
\label{fig:aimed}
\end{figure}

\emph{Photons---}
The same applies for \cref{sec:photons}'s photonic dephasing: the transverse boost acts on each helicity qubit as a phase, a Bloch rotation about the helicity axis, so a pre-twisting source and the boost rotate about the same axis for every momentum and the phases add. A source applying the momentum-conditioned helicity phase $e^{i\lambda\kappa\,\xi(p;\zeta^*)}$ leaves the observer at rapidity $\zeta$ with effective phase $\xi(p;\zeta)-\kappa\,\xi(p;\zeta^*)$, and
\begin{equation}
\eta_{\rm eff}=\left\langle\cos 2\!\left[\xi(\zeta)-\kappa\,\xi(\zeta^*)\right]\right\rangle ,
\label{eq:photonaimed}
\end{equation}
with $\chi_{\rm ph}(\eta_{\rm eff})$ from \cref{eq:chiph}, again verified against the full conditioned-effect computation. At $\kappa=1$ and $\zeta=\zeta^*$ the observer finds $\eta_{\rm eff}=1$ and $\chi_{\rm ph}=6$; everything else said in \cref{sec:aimed}, including the $\tanh(\zeta/2)$ asymmetry, carries over (\cref{fig:photonaimed}).

\begin{figure}[tb]
\centering
\includegraphics[width=\columnwidth]{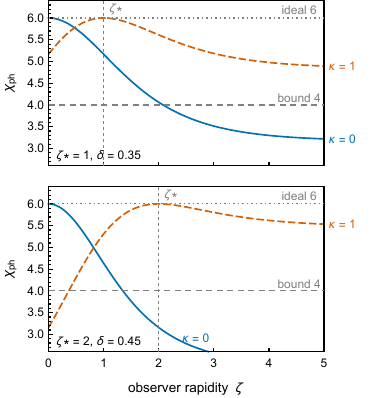}
\caption{Polarisation Peres-Mermin value against observer rapidity for an unaimed source ($\kappa=0$) and a source aimed at a transversely boosted frame $\zeta^*$ ($\kappa=1$), for $(\zeta^*=1,\delta=0.35)$ (top) and $(\zeta^*=2,\delta=0.45)$ (bottom). Both beam spreads exceed $\delta^*$, so the unaimed square fails at finite rapidity; the aimed square is ideal at $\zeta^*$.}
\label{fig:photonaimed}
\end{figure}

The photonic version also has the benefit of being potentially implementable. The phase $\xi$ is energy independent, so the compensating element acts on the angular spectrum alone, and its paraxial leading order, from \cref{eq:xiparaxial}, is a helicity-conditioned phase linear in one transverse wavevector component: a birefringent gradient element. Implementing only this linear term cancels the $\mathcal{O}(\theta)$ phase per momentum and leaves \cref{eq:xiparaxial}'s second-order term. Averaging $\cos2\xi$ over the Gaussian angular spectrum, with $\langle\theta^4\rangle=8\delta^4$ and $\langle\sin^2\varphi\cos^2\varphi\rangle=1/8$, gives
\begin{equation}
1-\eta_{\rm eff}\Big|_{\zeta=\zeta^*}\to\tfrac{1}{2}\tanh^4\!\left(\tfrac{\zeta^*}{2}\right)\delta^4
\label{eq:linres}
\end{equation}
for small $\delta$; a log-log fit of the full residual over $\delta\leq0.15$ gives an exponent of $4.09$, the excess over $4$ coming from the $\mathcal{O}(\theta^3)$ terms of \cref{eq:xiparaxial}. At $\zeta^*=1$ the shortfall $6-\chi_{\rm ph}$ is $7\times10^{-4}$ at $\delta=0.2$, $5\times10^{-3}$ at $\delta=0.3$, and $0.07$ at $\delta=0.5$, beyond the critical divergence, against a margin of $2$ above the noncontextuality bound. Aiming the photonic square is standard polarisation optics to $\delta^4$ accuracy, and it moves the critical-divergence constraint from the source frame to a neighbourhood of $\zeta^*$. The compensating element is therefore within reach; the relative motion is not, and a $\zeta^*$ of order one remains far outside anything a laboratory source and detector can be put into.

\section{Emulating the boost channel in polarisation optics}\label{sec:emulation}

\cref{tab:hierarchy}'s thresholds lie beyond any relative motion a laboratory can arrange, but the channel that produces them does not. \Cref{eq:channel} gives a rotation of a qubit by an angle $\theta_W$ about an axis $\hat{n}(\varphi)$ that turns with the azimuth of a second degree of freedom, followed by a trace over that degree of freedom. Polarisation optics has an element that does exactly this to the polarisation of a light beam, with the transverse position of the beam in the role of the momentum: the $q$-plate \cite{marrucci2006optical}, a birefringent plate whose optic axis makes the angle $\alpha(\varphi)=q\varphi+\alpha_0$ with a reference direction at azimuth $\varphi$ around the beam axis.

\emph{The map---}
A waveplate of retardation $\delta$ with its axis at angle $\alpha$ acts on the polarisation Jones vector as
\begin{equation}
\hat{J}(\delta,\alpha)=\exp\left[-i\tfrac{\delta}{2}\left(\cos2\alpha\,\hat{\sigma}_z+\sin2\alpha\,\hat{\sigma}_x\right)\right]
\label{eq:jones}
\end{equation}
in the horizontal/vertical polarisation basis, a rotation of the Poincar\'e sphere by $\delta$ about an axis in the equatorial plane of linear polarisations, at angle $2\alpha$ from the $\hat{\sigma}_z$ direction. For $q=\tfrac12$ the axis angle is $2\alpha=\varphi+2\alpha_0$, so the plate rotates the polarisation at azimuth $\varphi$ by a fixed angle about an axis that turns once with $\varphi$: this is \cref{eq:rotform} with $\theta_W\to\delta$ and the circular-polarisation direction $\hat{\sigma}_y$ in the role of the boost axis. A detector that does not resolve the transverse position traces the azimuth out, and for any beam whose intensity is azimuthally symmetric the average of \cref{eq:jones} over $\varphi$ is
\begin{equation}
M=\mathrm{diag}(1-\Delta,\,1-\Delta,\,1-2\Delta),\qquad \Delta=\sin^2(\delta/2),
\label{eq:qplatechannel}
\end{equation}
on the components $(\hat{\sigma}_z,\hat{\sigma}_x,\hat{\sigma}_y)$: the massive channel \cref{eq:etas}, exactly and without paraxial approximation, with the boost axis as the circular axis and the two linear axes transverse. The retardation of liquid-crystal $q$-plates is tuned electrically over $[0,\pi]$ in milliseconds \cite{piccirillo2010photon,dambrosio2015arbitrary}, so a voltage takes $\Delta$ from $0$ to $1$, the whole range of a random-rotation channel and well beyond the values \cref{sec:channel}'s packets reach. A uniform retardation reproduces a packet whose Wigner angle is the same at every momentum, so the higher moments are $\Delta_j=\Delta_1^j$; a radially varying retardation profile $\delta(r)$ across a beam of known intensity profile would supply the moments of a genuine packet, but is not needed for the qubit results, which depend on $\Delta_1$ alone. \cref{tab:hierarchy}'s certificate thresholds become retardations: $28.0^\circ$ and $31.7^\circ$ for the two orientations of the square, $47.0^\circ$ for CHSH with a plate in each arm, $54.0^\circ$ for the GHZ gate with a plate in each of three arms, and $65.5^\circ$ for CHSH with a plate in one arm.

The Lorentz group has appeared in polarisation optics before, as the action of attenuators and phase plates on Jones vectors \cite{han1999wigner}; here instead the boost enters only through the momentum-conditioned rotation and the trace, and the plate's transverse coordinate stands in for the momentum rather than for a Lorentz transformation.

\emph{What the emulation shows---}
Send a polarisation-entangled photon pair through a $q$-plate in each arm and measure CHSH with ordinary analysers after the plates. The analysers act as \cref{sec:bell}'s comoving devices, the recorded value is
\begin{equation}
S=2\sqrt{2}\,(1-\Delta)^2=2\sqrt{2}\cos^4(\delta/2),
\label{eq:chshqplate}
\end{equation}
and the classical bound is crossed at $\delta=47.0^\circ$, with the one-arm variant $2\sqrt2\cos^2(\delta/2)$ crossing at $65.5^\circ$. A state-independent test in the same setup behaves differently. With the two qubits carried by the polarisation and the path of a single photon, as in the photonic Peres-Mermin experiment of Amselem \emph{et al.}~\cite{amselem2009state}, a plate acting on the polarisation qubit leaves the sequential value at $6$ for every retardation, as \cref{prop:invariance} requires of sharp comoving measurements, while the CHSH value between the same two qubits falls as $2\sqrt2\cos^2(\delta/2)$. \cref{tab:hierarchy}'s two rows that differ similarly, the state-dependent certificate that a comoving detector records and the state-independent one that it does not lose, are then visible side by side as the retardation is varied. \cref{sec:pm}'s GHZ gate is emulated the same way: a three-photon polarisation GHZ state \cite{bouwmeester1999observation}, a plate in each arm and analysers after them in the two linear-polarisation bases, which stand in for $\hat{X}$ and $\hat{Y}$ and are both transverse to the plate's circular axis. Being state dependent, its success probability is recorded, $p=\tfrac12[1+\cos^6(\delta/2)]$ from \cref{eq:ghzprob}, and crosses Raussendorf's bound $3/4$ at $\delta=54.0^\circ$.

\cref{sec:aimed}'s aimed source is a second $q$-plate at the source. Its optic-axis pattern is set with $\alpha_0$ advanced by $\pi/2$, which reverses $\hat{n}(\varphi)$, and its retardation to $\kappa\delta^*$; if the two plates are imaged onto each other so that the azimuth is preserved between them, the two rotations at each azimuth share an axis and their angles subtract, giving $\Delta_{\rm eff}=\sin^2[(\delta-\kappa\delta^*)/2]$ as in \cref{eq:aimed}. At $\kappa=1$ the CHSH value $2\sqrt2\cos^4[(\delta-\delta^*)/2]$ is maximal for the detector at retardation $\delta^*$ and reduced to $2\sqrt2\cos^4(\delta^*/2)$ at $\delta=0$: this is \cref{fig:aimed}'s relocation, with the retardation in place of the observer's rapidity.

\emph{What it does not show---}
\cref{eq:chiDelta}'s sequential value $\chi(\Delta)$ is the value an observer attributes to the lab devices in the spin-only account, implemented with the L\"uders instruments of the unsharp effects $\tilde{\Pi}_\pm$. No element of the setup above implements those instruments: analysers before the plate are the lab devices and record $6$, and analysers after it are comoving devices and record $6$ as well. One could realise the instruments directly, with ancilla-assisted measurements of adjustable strength, but the experiment would then test the contextuality of unsharp measurements in the sense of Krishna \emph{et al.}~\cite{krishna2017deriving}, with the plate supplying only the value of $\Delta$. The emulation therefore makes \cref{sec:bell,sec:aimed} operational, which is where the recorded frequencies change, and leaves the frame dependence of the KS account where \cref{sec:pm} put it, in the description. \cref{sec:photons}'s transverse photonic channel, which dephases about the circular axis by an azimuth-dependent angle, has no standard optical element, though a spatially patterned polarisation rotator would provide this.

\section{Momentum as a context-selecting environment}\label{sec:context}

The transition from \cref{eq:square} to \cref{eq:noisyeffects} is an example of the mechanism analysed in Ref.~\cite{hance2026external}, where an environment selects the measurement context outcome by outcome. There, the environment-conditioned outcome states were obtained by applying the initial environmental state to the joint outcome, $\ket{\lambda(m)}_S={}_E\braket{\varphi_{\rm init}}{m}_{ES}$, and a family of environment-controlled unitaries $\hat{U}_x$ turned quantum fluctuations of the control into fluctuations of the context.

The boosted spin measurement provides a physical mechanism for this. The environment is the momentum degree of freedom of the particle itself; the initial environmental state is the wave packet $\psi(p)$; the controlled unitary is the Wigner rotation $\hat{D}(p)$ of \cref{eq:rotform}, applied to the spin conditionally on the momentum; and the coupling strength is set by the rapidity, through $\theta_W$. In the rest frame the coupling is absent and the contexts of the square are the familiar sharp ones. In the boosted frame each outcome of a nominally fixed spin measurement is correlated with the momentum fluctuations, so outcomes of a single device no longer belong to a single sharp context, and the counterfactual reasoning that assigns one context per device fails, in line with Ref.~\cite{hance2026external}'s argument that contexts cannot in general be identified with complete orthogonal bases. The loss of the Peres-Mermin violation beyond the critical curve is then not contextuality disappearing from the world (the full-space statistics still show it), but the boosted observer's spin-sector measurements ceasing to be the kind of object the KS argument speaks about.

Let us compare this to Ref.~\cite{camposdelgado2026frame}'s scenario. There, the boost acted on a fixed set of states, and frame dependence appeared because linear dependence of spin states is not preserved by \cref{eq:channel}. Here, states are irrelevant, since \cref{eq:chiDelta} is state independent, and the frame dependence sits entirely in the observables, through the failure of the operator identities. The two results are complementary faces of one fact, already visible in the non-covariance of the reduced spin state \cite{peres2002quantum,peres2004quantum}: the spin sector of a massive particle is not a Lorentz-covariant subsystem, so notions defined on it, via states or via observables, need not transform covariantly. For generalised noncontextuality, incompatibility is neither necessary nor sufficient \cite{selby2023contextuality}; for KS proofs it remains necessary. The aligned square keeps every within-context commutation and every cross-context noncommutation at every $\Delta$, and still loses its certificate, so the boost removes measurement sharpness rather than incompatibility.

\Cref{sec:pm} and the $J_z$-diagonal rule of \cref{sec:arbitrary} refine this: the environment goes beyond just adding noise, to selecting which families of spin observables still form commuting contexts at all, with \cref{sec:arbitrary}'s Free Will triple being an exception preserved by its particular operator algebra. Context selection by the momentum environment is visible at three levels: which contexts remain for a boosted observer, how sharp their outcomes are, and how much contextuality their statistics retain. \Cref{prop:planar} gives the second level ray by ray: each ray loses sharpness at a rate set by its planar spin variance.

\section{Discussion}\label{sec:discussion}

We have shown that KS contextuality, like Spekkens' contextuality, is Lorentz invariant on the full Hilbert space of a massive particle but frame dependent when restricted to spin. As \cref{sec:context} describes, the same channel is behind both, read in two pictures: it mixes the states for Spekkens' notion and de-sharpens the observables for the KS notion. The state-independent Peres-Mermin value falls from $6$ as \cref{eq:chiDelta}, crossing the classical bound along a critical curve in the $(\sigma/m,\zeta)$ plane, with the transverse orientation the most robust choice; the master formula \cref{eq:master} extends the channel to every spin, and \cref{prop:planar} gives the rate at which each ray loses sharpness, its planar spin variance; and \cref{tab:hierarchy}'s certificates fail in a fixed order, the Free Will Theorem's certainty first, then the qutrit ray set because its rays have the largest planar variances, then the square, then two-wing CHSH and the GHZ computation, and one-wing CHSH last. What a device-independent protocol concludes from an observed violation is unaffected, though how much violation is available is not.

The qubit results are governed by the single scalar $\Delta$, whose saturation value $\sigma^2/2m^2$ makes the violation robust for any source with $\sigma\ll m$: as with the frame dependence of spin entropy \cite{peres2002quantum}, observing the loss requires momentum spreads on the order of the mass, i.e.\ wave packets localised near the Compton scale. Satellite-scale relative velocities give $\Delta$ many orders of magnitude below every threshold, free massive-particle sources have $\sigma\ll m$, and photonic links are immune for motion along the beam. One system sets a suggestive scale: the momentum spread of an inner-shell electron in a hydrogen-like ion is of order $Z\alpha\,m$, so heavy ions sit at $\sigma/m\sim 0.5$--$0.7$ ($Z\alpha\approx 0.67$ for hydrogen-like uranium). That range activates the SPIN axiom, both Yu-Oh rows and the aligned Peres-Mermin square, and the transverse square from $0.503\,m$ up, but neither CHSH row: for those the packet would have to be localised at or inside the Compton wavelength. Our free-packet analysis does not carry over verbatim to bound states, but still provides an illustrative scale, and Chen \emph{et al.}'s fixed-frame violation for the relativistic hydrogen atom \cite{chen2013quantum} shows that such systems would potentially allow these tests.

Of the quantities in \cref{tab:hierarchy} the Free Will Theorem failure rate degrades fastest, being linear in $\Delta_1$ with slope $16/3$: a spread of $\sigma\approx 0.045\,m$ puts the spin-only account half a percent away from the $101$ rule at high rapidity. As in \cref{sec:pm}, that is a statement about the account and not about recorded frequencies, which give $101$ at every momentum. For photons the relevant scale is instead an angular spread, $\delta^*\approx0.28$~rad, reachable with strong focusing. One caveat applies to the massive numbers themselves: a Gaussian packet with $\sigma\gtrsim m$ is localised at or inside the Compton wavelength, so the single-particle positive-energy description is at the edge of its validity just where the thresholds live.

None of the massive thresholds are reachable by motion, but every qubit threshold in \cref{tab:hierarchy} is reachable by retardation. \Cref{sec:emulation}'s $q$-plate applies the massive channel to photon polarisation with no approximation, at $28^\circ$ to $66^\circ$ of retardation for the five qubit certificates, and a second plate at the source aims it. The experiment we would run first is the CHSH relocation: a source plate at $\delta^*$ and a detector plate swept through $\delta$, with the violation maximal at $\delta=\delta^*$ rather than at $\delta=0$. The same setup shows the state-independent square holding at $6$ while CHSH falls, and with three photons shows the GHZ gate losing its certificate at $54^\circ$, but it cannot show the spin-only value $\chi(\Delta)$, which needs the L\"uders instruments of unsharp effects that no plate supplies.

The results order the carriers of contextuality by frame robustness. Photonic polarisation under collinear boosts is the one case protected. Next come massive qubit pairs, protected for $\sigma\lesssim0.42$--$0.50\,m$; transversely boosted photons; and last the Yu-Oh set, with every higher-spin test after it: by \cref{eq:vwindow} the rays of a spin-$s$ set lose sharpness at a rate between $s^{2/3}$ and $s^2$ times $\Delta_1$, so thresholds fall at least as fast as $s^{-2/3}$ and at most as fast as $s^{-2}$, with a random ray set at the fast end. The KS theorem's original dimension, $d=3$, is therefore where its relativistic certification is weakest, while the qubit pair, where the theorem needs the two-particle square rather than rays, is the most robust massive carrier.

Aimed sources recast this frame-dependence. The sharpness of the spin-sector contexts is a resource that sits in some frame, and \cref{eq:aimed} shows the source choosing which; \cref{tab:hierarchy}'s hierarchy then measures how far from that frame each certificate keeps working. Open directions include how aiming interacts with device-independent protocols when the two wings of a Bell pair are aimed at different frames, the bound-state version of the analysis for heavy ions, repeating the construction for other choices of relativistic spin operator \cite{bauke2014relativistic,saldanha2012physical}, and three raised by the additions above: the smallest average planar variance an orthogonality graph permits, which fixes the exponent with which KS thresholds fall with spin; how far the spin-only account depends on the instrument chosen for the unsharp effects, since the L\"uders instruments used here are the canonical choice but not the only one; and a patterned polarisation rotator that would emulate the transverse photonic channel as the $q$-plate emulates the massive one.

\textit{Data availability -} A Wolfram Mathematica notebook reproducing all computations in this paper is openly available on Zenodo~\cite{HanceCamposDelgado2026Notebook}.

\textit{Acknowledgements -} JRH acknowledges support from a Royal Society Research Grant (RG/R1/251590), an EPSRC Mathematical Sciences Small Grant (UKRI3647), and from their EPSRC Quantum Technologies Career Acceleration Fellowship (UKRI1217).

\bibliographystyle{unsrturl}
\bibliography{ref}

@article{kochen1967problem,
  title={The problem of hidden variables in quantum mechanics},
  author={Kochen, Simon and Specker, Ernst P.},
  journal={J. Math. Mech.},
  volume={17},
  number={1},
  pages={59--87},
  year={1967},
  doi={10.1512/iumj.1968.17.17004}
}

@article{bell1966problem,
  title={On the problem of hidden variables in quantum mechanics},
  author={Bell, John S.},
  journal={Rev. Mod. Phys.},
  volume={38},
  pages={447--452},
  year={1966},
  doi={10.1103/RevModPhys.38.447}
}

@article{spekkens2005contextuality,
  title={Contextuality for preparations, transformations, and unsharp measurements},
  author={Spekkens, Robert W.},
  journal={Phys. Rev. A},
  volume={71},
  pages={052108},
  year={2005},
  doi={10.1103/PhysRevA.71.052108}
}

@article{camposdelgado2026frame,
  title={Frame dependence of Spekkens' contextuality for relativistic spin systems},
  author={Campos Delgado, Ruben},
  journal={Phys. Rev. A},
  volume={113},
  pages={012202},
  year={2026},
  doi={10.1103/crqz-jcsh}
}

@article{peres1990incompatible,
  title={Incompatible results of quantum measurements},
  author={Peres, Asher},
  journal={Phys. Lett. A},
  volume={151},
  number={3-4},
  pages={107--108},
  year={1990},
  doi={10.1016/0375-9601(90)90172-K}
}

@article{mermin1990simple,
  title={Simple unified form for the major no-hidden-variables theorems},
  author={Mermin, N. David},
  journal={Phys. Rev. Lett.},
  volume={65},
  pages={3373--3376},
  year={1990},
  doi={10.1103/PhysRevLett.65.3373}
}

@article{cabello2008experimentally,
  title={Experimentally testable state-independent quantum contextuality},
  author={Cabello, Ad{\'a}n},
  journal={Phys. Rev. Lett.},
  volume={101},
  pages={210401},
  year={2008},
  doi={10.1103/PhysRevLett.101.210401}
}

@article{budroni2022kochen,
  title={Kochen-Specker contextuality},
  author={Budroni, Costantino and Cabello, Ad{\'a}n and G{\"u}hne, Otfried and Kleinmann, Matthias and Larsson, Jan-{\AA}ke},
  journal={Rev. Mod. Phys.},
  volume={94},
  pages={045007},
  year={2022},
  doi={10.1103/RevModPhys.94.045007}
}

@article{guhne2010compatibility,
  title={Compatibility and noncontextuality for sequential measurements},
  author={G{\"u}hne, Otfried and Kleinmann, Matthias and Cabello, Ad{\'a}n and Larsson, Jan-{\AA}ke and Kirchmair, Gerhard and Z{\"a}hringer, Florian and Gerritsma, Rene and Roos, Christian F.},
  journal={Phys. Rev. A},
  volume={81},
  pages={022121},
  year={2010},
  doi={10.1103/PhysRevA.81.022121}
}

@article{wigner1939unitary,
  title={On unitary representations of the inhomogeneous {L}orentz group},
  author={Wigner, Eugene P.},
  journal={Ann. Math.},
  volume={40},
  number={1},
  pages={149--204},
  year={1939},
  doi={10.2307/1968551}
}

@book{weinberg1995quantum,
  title={The Quantum Theory of Fields, Volume 1: Foundations},
  author={Weinberg, Steven},
  year={1995},
  publisher={Cambridge University Press},
  address={Cambridge},
  doi={10.1017/CBO9781139644167}
}

@article{peres2002quantum,
  title={Quantum entropy and special relativity},
  author={Peres, Asher and Scudo, Petra F. and Terno, Daniel R.},
  journal={Phys. Rev. Lett.},
  volume={88},
  pages={230402},
  year={2002},
  doi={10.1103/PhysRevLett.88.230402}
}

@article{peres2004quantum,
  title={Quantum information and relativity theory},
  author={Peres, Asher and Terno, Daniel R.},
  journal={Rev. Mod. Phys.},
  volume={76},
  pages={93--123},
  year={2004},
  doi={10.1103/RevModPhys.76.93}
}

@article{gingrich2002quantum,
  title={Quantum entanglement of moving bodies},
  author={Gingrich, Robert M. and Adami, Christoph},
  journal={Phys. Rev. Lett.},
  volume={89},
  pages={270402},
  year={2002},
  doi={10.1103/PhysRevLett.89.270402}
}

@article{czachor1997einstein,
  title={Einstein-{P}odolsky-{R}osen-{B}ohm experiment with relativistic massive particles},
  author={Czachor, Marek},
  journal={Phys. Rev. A},
  volume={55},
  pages={72--77},
  year={1997},
  doi={10.1103/PhysRevA.55.72}
}

@article{terashima2003einstein,
  title={Einstein-{P}odolsky-{R}osen correlation seen from moving observers},
  author={Terashima, Hiroaki and Ueda, Masahito},
  journal={Quantum Inf. Comput.},
  volume={3},
  number={3},
  pages={224--228},
  year={2003},
  doi={10.26421/QIC3.3-4}
}

@article{terashima2003relativistic,
  title={Relativistic {E}instein-{P}odolsky-{R}osen correlation and {B}ell's inequality},
  author={Terashima, Hiroaki and Ueda, Masahito},
  journal={Int. J. Quantum Inf.},
  volume={1},
  number={1},
  pages={93--114},
  year={2003},
  doi={10.1142/S0219749903000061}
}

@misc{churilov2026fixed,
  title={Fixed-boost {W}igner noise: strict trace-distance contraction without quantum degradability},
  author={Churilov, Maxim V.},
  year={2026},
  eprint={2607.12994},
  archivePrefix={arXiv},
  primaryClass={quant-ph}
}

@article{chen2013quantum,
  title={Quantum contextuality for a relativistic spin-1/2 particle},
  author={Chen, Jing-Ling and Su, Hong-Yi and Wu, Chunfeng and Deng, Dong-Ling and Cabello, Ad{\'a}n and Kwek, Leong Chuan and Oh, Choo Hiap},
  journal={Phys. Rev. A},
  volume={87},
  pages={022109},
  year={2013},
  doi={10.1103/PhysRevA.87.022109}
}

@article{meyer1999finite,
  title={Finite precision measurement nullifies the {K}ochen-{S}pecker theorem},
  author={Meyer, David A.},
  journal={Phys. Rev. Lett.},
  volume={83},
  pages={3751--3754},
  year={1999},
  doi={10.1103/PhysRevLett.83.3751}
}

@article{kent1999noncontextual,
  title={Noncontextual hidden variables and physical measurements},
  author={Kent, Adrian},
  journal={Phys. Rev. Lett.},
  volume={83},
  pages={3755--3757},
  year={1999},
  doi={10.1103/PhysRevLett.83.3755}
}

@article{clifton2000simulating,
  title={Simulating quantum mechanics by non-contextual hidden variables},
  author={Clifton, Rob and Kent, Adrian},
  journal={Proc. R. Soc. A},
  volume={456},
  pages={2101--2114},
  year={2000},
  doi={10.1098/rspa.2000.0604}
}

@article{barrett2004non,
  title={Non-contextuality, finite precision measurement and the {K}ochen-{S}pecker theorem},
  author={Barrett, Jonathan and Kent, Adrian},
  journal={Stud. Hist. Philos. Sci. B},
  volume={35},
  number={2},
  pages={151--176},
  year={2004},
  doi={10.1016/j.shpsb.2003.10.003}
}

@article{spekkens2014status,
  title={The status of determinism in proofs of the impossibility of a noncontextual model of quantum theory},
  author={Spekkens, Robert W.},
  journal={Found. Phys.},
  volume={44},
  pages={1125--1155},
  year={2014},
  doi={10.1007/s10701-014-9833-x}
}

@article{selby2023contextuality,
  title={Contextuality without incompatibility},
  author={Selby, John H. and Schmid, David and Wolfe, Elie and Sainz, Ana Bel{\'e}n and Kunjwal, Ravi and Spekkens, Robert W.},
  journal={Phys. Rev. Lett.},
  volume={130},
  pages={230201},
  year={2023},
  doi={10.1103/PhysRevLett.130.230201}
}

@article{hance2026external,
  title={External quantum fluctuations select measurement contexts},
  author={Hance, Jonte R. and Ji, Ming and Matsushita, Tomonori and Hofmann, Holger F.},
  journal={Quantum},
  volume={10},
  pages={2106},
  year={2026},
  doi={10.22331/q-2026-05-13-2106}
}

@article{saldanha2012physical,
  title={Physical interpretation of the {W}igner rotations and its implications for relativistic quantum information},
  author={Saldanha, Pablo L. and Vedral, Vlatko},
  journal={New J. Phys.},
  volume={14},
  pages={023041},
  year={2012},
  doi={10.1088/1367-2630/14/2/023041}
}

@article{bauke2014relativistic,
  title={What is the relativistic spin operator?},
  author={Bauke, Heiko and Ahrens, Sven and Keitel, Christoph H. and Grobe, Rainer},
  journal={New J. Phys.},
  volume={16},
  pages={043012},
  year={2014},
  doi={10.1088/1367-2630/16/4/043012}
}

@article{horodecki1995violating,
  title={Violating {B}ell inequality by mixed spin-1/2 states: necessary and sufficient condition},
  author={Horodecki, Ryszard and Horodecki, Pawe{\l} and Horodecki, Micha{\l}},
  journal={Phys. Lett. A},
  volume={200},
  number={5},
  pages={340--344},
  year={1995},
  doi={10.1016/0375-9601(95)00214-N}
}

@article{brunner2014bell,
  title={Bell nonlocality},
  author={Brunner, Nicolas and Cavalcanti, Daniel and Pironio, Stefano and Scarani, Valerio and Wehner, Stephanie},
  journal={Rev. Mod. Phys.},
  volume={86},
  pages={419--478},
  year={2014},
  doi={10.1103/RevModPhys.86.419}
}

@article{acin2007device,
  title={Device-independent security of quantum cryptography against collective attacks},
  author={Ac{\'\i}n, Antonio and Brunner, Nicolas and Gisin, Nicolas and Massar, Serge and Pironio, Stefano and Scarani, Valerio},
  journal={Phys. Rev. Lett.},
  volume={98},
  pages={230501},
  year={2007},
  doi={10.1103/PhysRevLett.98.230501}
}

@article{conway2006free,
  title={The free will theorem},
  author={Conway, John H. and Kochen, Simon},
  journal={Found. Phys.},
  volume={36},
  number={10},
  pages={1441--1473},
  year={2006},
  doi={10.1007/s10701-006-9068-6}
}

@article{yu2012state,
  title={State-independent proof of {K}ochen-{S}pecker theorem with 13 rays},
  author={Yu, Sixia and Oh, Choo Hiap},
  journal={Phys. Rev. Lett.},
  volume={108},
  pages={030402},
  year={2012},
  doi={10.1103/PhysRevLett.108.030402}
}

@misc{jokinen2026generalised,
  title={Generalised contextuality of continuous variable quantum theory can be revealed with a single projective measurement},
  author={Jokinen, Pauli and Weilenmann, Mirjam and Pl{\'a}vala, Martin and Pellonp{\"a}{\"a}, Juha-Pekka and Kiukas, Jukka and Uola, Roope},
  year={2026},
  eprint={2601.14067},
  archivePrefix={arXiv},
  primaryClass={quant-ph}
}

@article{cabello2016thirteen,
  title={Quantum state-independent contextuality requires 13 rays},
  author={Cabello, Ad{\'a}n and Kleinmann, Matthias and Portillo, Jos{\'e} R.},
  journal={J. Phys. A},
  volume={49},
  number={38},
  pages={38LT01},
  year={2016},
  doi={10.1088/1751-8113/49/38/38LT01}
}

@article{krishna2017deriving,
  title={Deriving robust noncontextuality inequalities from algebraic proofs of the {K}ochen-{S}pecker theorem: the {P}eres-{M}ermin square},
  author={Krishna, Anirudh and Spekkens, Robert W. and Wolfe, Elie},
  journal={New J. Phys.},
  volume={19},
  pages={123031},
  year={2017},
  doi={10.1088/1367-2630/aa9168}
}

@article{cabello2002finite,
  title={Finite-precision measurement does not nullify the {K}ochen-{S}pecker theorem},
  author={Cabello, Ad{\'a}n},
  journal={Phys. Rev. A},
  volume={65},
  pages={052101},
  year={2002},
  doi={10.1103/PhysRevA.65.052101}
}

@article{bartlett2005relativistically,
  title={Relativistically invariant quantum information},
  author={Bartlett, Stephen D. and Terno, Daniel R.},
  journal={Phys. Rev. A},
  volume={71},
  pages={012302},
  year={2005},
  doi={10.1103/PhysRevA.71.012302}
}

@article{gingrich2003entangled,
  title={Entangled light in moving frames},
  author={Gingrich, Robert M. and Bergou, Attila J. and Adami, Christoph},
  journal={Phys. Rev. A},
  volume={68},
  pages={042102},
  year={2003},
  doi={10.1103/PhysRevA.68.042102}
}

@article{ho2020noisy,
  title={Noisy preprocessing facilitates a photonic realization of device-independent quantum key distribution},
  author={Ho, Melvyn and Sekatski, Pavel and Tan, Ernest Y.-Z. and Renner, Renato and Bancal, Jean-Daniel and Sangouard, Nicolas},
  journal={Phys. Rev. Lett.},
  volume={124},
  pages={230502},
  year={2020},
  doi={10.1103/PhysRevLett.124.230502}
}

@article{hermens2014conway,
  title={{C}onway-{K}ochen and the finite precision loophole},
  author={Hermens, Ronnie},
  journal={Found. Phys.},
  volume={44},
  number={10},
  pages={1038--1048},
  year={2014},
  doi={10.1007/s10701-014-9827-8}
}

@article{cator2014constraints,
  title={Constraints on determinism: {B}ell versus {C}onway-{K}ochen},
  author={Cator, Eric and Landsman, Klaas},
  journal={Found. Phys.},
  volume={44},
  number={7},
  pages={781--791},
  year={2014},
  doi={10.1007/s10701-014-9815-z}
}

@misc{dutra2026contextuality,
  title={How quantum contextuality disappears in the classical limit},
  author={Dutra, Arthur C. R. and Baldij{\~a}o, Roberto D. and Terra Cunha, Marcelo},
  year={2026},
  eprint={2604.24735},
  archivePrefix={arXiv},
  primaryClass={quant-ph}
}

@article{streiter2021relativistic,
  title={Relativistic {B}ell test within quantum reference frames},
  author={Streiter, Lucas F. and Giacomini, Flaminia and Brukner, {\v{C}}aslav},
  journal={Phys. Rev. Lett.},
  volume={126},
  pages={230403},
  year={2021},
  doi={10.1103/PhysRevLett.126.230403}
}

@article{lindner2005focusing,
  title={The effect of focusing on polarization qubits},
  author={Lindner, Netanel H. and Terno, Daniel R.},
  journal={J. Mod. Opt.},
  volume={52},
  number={8},
  pages={1177--1188},
  year={2005},
  doi={10.1080/09500340512331327570}
}

@book{Halpern1968,
  author       = {Francis R. Halpern},
  title        = "{Special Relativity and Quantum Mechanics}",
  publisher    = {Prentice-Hall},
  address      = {Englewood Cliffs, New Jersey},
  year         = {1968},
  isbn         = {0138271135, 9780138271138},
}

@article{Raussendorf:2001zim,
    author = "Raussendorf, Robert and Briegel, Hans J.",
    title = "{A One-Way Quantum Computer}",
    doi = "10.1103/PhysRevLett.86.5188",
    journal = "Phys. Rev. Lett.",
    volume = "86",
    number = "22",
    pages = "5188",
    year = "2001"
}

@article{Raussendorf:2012tof,
    author = "Raussendorf, Robert and Wei, Tzu-Chieh",
    title = "{Quantum Computation by Local Measurement}",
    eprint = "1208.0041",
    archivePrefix = "arXiv",
    primaryClass = "quant-ph",
    doi = "10.1146/annurev-conmatphys-020911-125041",
    journal = "Ann. Rev. Condensed Matter Phys.",
    volume = "3",
    number = "1",
    pages = "239--261",
    year = "2012"
}

@article{Raussendorf:2013wih,
    author = "Raussendorf, Robert",
    title = "{Contextuality in measurement-based quantum computation}",
    eprint = "0907.5449",
    archivePrefix = "arXiv",
    primaryClass = "quant-ph",
    doi = "10.1103/PhysRevA.88.022322",
    journal = "Phys. Rev. A",
    volume = "88",
    number = "2",
    pages = "022322",
    year = "2013"
}

@article{Anders:2009zzs,
    author = "Anders, Janet and Browne, Dan E.",
    title = "{Computational Power of Correlations}",
    eprint = "0805.1002",
    archivePrefix = "arXiv",
    primaryClass = "quant-ph",
    doi = "10.1103/PhysRevLett.102.050502",
    journal = "Phys. Rev. Lett.",
    volume = "102",
    number = "5",
    pages = "050502",
    year = "2009"
}

@article{marrucci2006optical,
  title={Optical spin-to-orbital angular momentum conversion in inhomogeneous anisotropic media},
  author={Marrucci, L. and Manni, C. and Paparo, D.},
  journal={Phys. Rev. Lett.}, volume={96}, pages={163905}, year={2006},
  doi={10.1103/PhysRevLett.96.163905}}

@article{piccirillo2010photon,
  title={Photon spin-to-orbital angular momentum conversion via an electrically tunable $q$-plate},
  author={Piccirillo, Bruno and D'Ambrosio, Vincenzo and Slussarenko, Sergei and Marrucci, Lorenzo and Santamato, Enrico},
  journal={Appl. Phys. Lett.}, volume={97}, pages={241104}, year={2010},
  doi={10.1063/1.3527083}}

@article{dambrosio2015arbitrary,
  title={Arbitrary, direct and deterministic manipulation of vector beams via electrically-tuned $q$-plates},
  author={D'Ambrosio, Vincenzo and Baccari, Flavio and Slussarenko, Sergei and Marrucci, Lorenzo and Sciarrino, Fabio},
  journal={Sci. Rep.}, volume={5}, pages={7840}, year={2015},
  doi={10.1038/srep07840}}

@article{han1999wigner,
  title={Wigner rotations and Iwasawa decompositions in polarization optics},
  author={Han, D. and Kim, Y. S. and Noz, M. E.},
  journal={Phys. Rev. E}, volume={60}, pages={1036}, year={1999},
  doi={10.1103/PhysRevE.60.1036}}

@article{amselem2009state,
  title={State-independent quantum contextuality with single photons},
  author={Amselem, Elias and R{\aa}dmark, Magnus and Bourennane, Mohamed and Cabello, Ad{\'a}n},
  journal={Phys. Rev. Lett.}, volume={103}, pages={160405}, year={2009},
  doi={10.1103/PhysRevLett.103.160405}}

@article{nielsen2002simple,
  title={A simple formula for the average gate fidelity of a quantum dynamical operation},
  author={Nielsen, Michael A.},
  journal={Phys. Lett. A}, volume={303}, pages={249--252}, year={2002},
  doi={10.1016/S0375-9601(02)01272-0}}

@article{he2011planar,
  title={Planar quantum squeezing and atom interferometry},
  author={He, Q. Y. and Peng, Shi-Guo and Drummond, P. D. and Reid, M. D.},
  journal={Phys. Rev. A}, volume={84}, pages={022107}, year={2011},
  doi={10.1103/PhysRevA.84.022107}}

@misc{bozzetto2026warringcontextualitiesprovably,
      title={Warring Contextualities - Provably Classical vs Provably Nonclassical},
      author={Enrico Bozzetto and Jonte R. Hance},
      year={2026},
      eprint={2604.14319},
      archivePrefix={arXiv},
      primaryClass={quant-ph},
}

@article{bouwmeester1999observation,
  title={Observation of three-photon Greenberger-Horne-Zeilinger entanglement},
  author={Bouwmeester, Dik and Pan, Jian-Wei and Daniell, Matthew and Weinfurter, Harald and Zeilinger, Anton},
  journal={Phys. Rev. Lett.}, volume={82}, pages={1345--1349}, year={1999},
  doi={10.1103/PhysRevLett.82.1345}}

@misc{HanceCamposDelgado2026Notebook,
  author = {Jonte R. Hance and Ruben Campos Delgado},
  title = {Code for ``Frame dependence of Kochen-Specker contextuality for relativistic spin systems'''},
  year = {2026},
  howpublished = {Zenodo},
  doi = {10.5281/zenodo.22942308},
  url = {https://doi.org/10.5281/zenodo.22942308}
}

\appendix

\section{A noncontextual model beyond the critical curve}\label{app:ncmodel}

\Cref{sec:pm} shows that the boosted Peres-Mermin value satisfies $\chi\leq4$ beyond the critical curve. That is one inequality; here we show that for the maximally mixed state nothing else obstructs a noncontextual assignment, so the whole boosted behaviour of the square admits a noncontextual model there.

Write $t_k$ for the boosted expectation of the $k$th triple, and note that on the maximally mixed state the single-observable expectations and the in-context pairs all vanish. Here $t_k=\varsigma_k g_k$, with $\varsigma=(1,1,1,1,1,-1)$ the signs of \cref{eq:chi} and $g_k$ the product of shrink factors along the $k$th triple in \cref{fig:square}.

A deterministic noncontextual assignment is a vector $v\in\{\pm1\}^9$, one value per observable, and it yields triple values whose product is $(\prod_i v_i)^2=+1$, since each observable sits in exactly one row and one column. The map $v\mapsto(t_k)$ has rank $5$ over $\mathrm{GF}(2)$, because the six rows of the incidence matrix sum to zero mod $2$; its image is therefore exactly the $32$ even-parity vertices of $\{\pm1\}^6$. Averaging over the $16$ sign flips in the kernel removes the singles and the in-context pairs without changing the $t_k$, so the reachable set is the convex hull of those vertices.

That hull has $44$ facets: the $12$ of the form $|t_k|\leq1$, and the $32$ of the form $\sum_k\varepsilon_kt_k\leq4$ with $\prod_k\varepsilon_k=-1$. Since $g_k\geq0$, the binding one is $\varepsilon=\varsigma$, which reads $\chi\leq4$; every other facet is slack. A linear program over the $512$ assignments confirms feasibility for $\chi\leq4$ and infeasibility above it.

\section{Choice of massless section}\label{app:section}

For a massless particle, the polarisation subspace is momentum-dependent, so a fixed helicity qubit across a wave packet exists only relative to a choice of section $L(p)$ carrying the reference momentum to $p$. \Cref{sec:photons} uses the geodesic section $\hat{R}_z(\varphi)\hat{R}_y(\theta)\hat{R}_z(-\varphi)\hat{B}_z$. The other common choice, the standard section $\hat{R}_z(\varphi)\hat{R}_y(\theta)\hat{B}_z$, gives a different answer.

The standard section contains a rotation about the momentum itself in the limit $\theta\to0$, so the basis it defines winds by $\varphi$ around the beam axis and is singular on the axis. A beam centred on that axis covers all $\varphi$, and the resulting phase spread is $\varphi$-uniform rather than $\theta$-controlled: the section shows complete dephasing for any angular spread $\delta$ once the boost's aberration exceeds $\delta$. At $\delta=0.2$ and $\zeta=3$ the two sections give $\eta=0.93$ and $\eta=0.02$.

The geodesic section, the minimal rotation carrying $\hat{z}$ to $\hat{p}$, is smooth on the axis and differs from the standard one only by a smooth $\varphi$-dependent phase away from it. It is the basis in which the transverse field of a paraxial beam is defined, and so the one a lens or waveplate acts in; the same identification underlies Lindner and Terno's treatment of focusing~\cite{lindner2005focusing}. The dephasing quoted in \cref{sec:photons} is therefore the one a polarisation analyser sees, and the complete dephasing of the standard section is an artefact of its basis rather than a physical effect. \Cref{prop:collinear} is unaffected either way, since a collinear boost leaves $\varphi$ unchanged.

\end{document}